\documentclass[onecolumn,a4paper,accepted=2026-03-03]{quantumarticle}
\pdfoutput=1
\usepackage[T1]{fontenc}
\usepackage[utf8]{inputenc}
\usepackage[english]{babel}
\usepackage{amsmath}
\usepackage{amsfonts}
\usepackage{amssymb}
\usepackage{amsthm}
\usepackage{expl3}
\usepackage{xparse}
\usepackage{calc}
\usepackage{array}
\usepackage{mathtools}
\usepackage[numbers,sort&compress]{natbib}
\usepackage{graphicx}
\usepackage{microtype}
\usepackage{hyperref}
\usepackage{slantsc}
\usepackage{xcolor}
\usepackage[shortlabels]{enumitem}
\usepackage{qcircuit}
\usepackage{breqn}
\usepackage{apptools}
\usepackage{aligned-overset}
\usepackage{subcaption}
\usepackage{bbm}
\usepackage{qcircuit}
\usepackage{algorithm}
\usepackage{algpseudocode}
\usepackage{varwidth}

\def\et al.{\textit{et~al.}}
\def\viz.{\textit{viz.}}

\makeatletter
\def\DontBreakHereButBreakBetweenItems{\@beginparpenalty=\@highpenalty\@itempenalty=-\@lowpenalty\@endparpenalty=-\@highpenalty}
\makeatother

\def\medglue{\vskip 6pt plus 2pt minus 6pt}
\def\bigglue{\vskip 12pt plus 4pt minus 12pt}

\NewDocumentCommand{\I}{}{\mathbf{i}}

\DeclarePairedDelimiter{\bra}{\langle}{\rvert}
\DeclarePairedDelimiter{\ket}{\lvert}{\rangle}
\NewDocumentCommand{\braket}{mm}{\langle#1\vert#2\rangle}
\NewDocumentCommand{\dyad}{m}{{\lvert#1\rangle\mathchoice{\mkern-1mu}{\mkern-1mu}{\mkern-2mu}{\mkern-3mu}\langle#1\rvert}}
\DeclarePairedDelimiterX{\ExpVal}[2]{\langle#1\vert}{\vert#1\rangle}{#2}
\DeclarePairedDelimiter{\abs}{\lvert}{\rvert}
\DeclarePairedDelimiter{\norm}{\lVert}{\rVert}
\DeclareMathOperator{\Tr}{Tr}
\DeclareMathOperator{\Rank}{Rank}

\undef{\Re}
\undef{\Im}
\DeclareMathOperator{\Re}{Re}
\DeclareMathOperator{\Im}{Im}

\NewDocumentCommand{\Qubits}{m}{\mathfrak{#1}}
\NewDocumentCommand{\QubitsN}{}{\Qubits{n}}
\NewDocumentCommand{\QubitsM}{}{\Qubits{m}}
\NewDocumentCommand{\StabNorm}{}{\mathbf{D}}

\DeclareMathOperator{\Query}{\mathbf{q}}
\NewDocumentCommand{\AbsQuery}{E{_}{{}}}{\mathbf{q}#1^{abs}}
\NewDocumentCommand{\RelQuery}{E{_}{{}}}{\mathbf{q}#1^{rel}}
\DeclareMathOperator{\Sample}{\mathbf{s}}
\DeclareMathOperator{\QNorm}{\mathbf{n}}

\NewDocumentCommand{\SQ}{o}{\IfValueTF{#1}{\ensuremath{\text{SQ}_{#1}}}{\text{SQ}}}
\NewDocumentCommand{\ASQ}{o}{\IfValueTF{#1}{\ensuremath{\text{ASQ}_{#1}}}{\text{ASQ}}}
\NewDocumentCommand{\AS}{o}{\IfValueTF{#1}{\ensuremath{\text{AS}_{#1}}}{\text{AS}}}

\NewDocumentCommand{\Probability}{}{\mathbb{P}}
\NewCommandCopy{\Prob}{\Probability}
\NewDocumentCommand{\Distribution}{t'}{\mathrm{\IfBooleanTF{#1}{p'}{p}}}
\NewDocumentCommand{\AltDistribution}{t'}{\mathrm{\IfBooleanTF{#1}{q'}{q}}}
\DeclareMathOperator*{\Expectation}{\mathbb{E}}
\DeclareMathOperator*{\Variance}{\textsc{var}}
\NewCommandCopy{\Var}{\Variance}

\NewDocumentCommand{\NegativeBinom}{s}{\mathrm{NB}\IfBooleanF{#1}{\qty}}

\NewDocumentCommand{\Complex}{}{\mathbb{C}}
\NewDocumentCommand{\Real}{}{\mathbb{R}}
\NewDocumentCommand{\Natural}{}{\mathbb{N}}

\NewDocumentCommand{\Defined}{}{\coloneq}
\NewDocumentCommand{\RDefined}{}{\eqcolon}

\DeclareMathOperator*{\BigO}{O}
\DeclareMathOperator{\Poly}{poly}
\DeclareMathOperator{\Polylog}{polylog}
\NewDocumentCommand{\Identity}{}{\mathrm{I}}

\NewDocumentCommand{\AbsStrut}{O{1pt}}{\rule[-#1]{0pt}{1em + #1}}

\NewDocumentCommand{\liff}{O{}}{\llap{$\overset{#1}{\iff}$ \hskip 1em}}%
\NewDocumentCommand{\limplies}{O{}}{\llap{$\overset{#1}{\implies}$ \hskip 1em}}%
\NewDocumentCommand{\limpliedby}{O{}}{\llap{$\overset{#1}{\impliedby}$ \hskip 1em}}%
\NewDocumentCommand{\lforall}{O{}}{{\llap{$\forall_{\IfValueTF{#1}{#1}{}}$ \hskip 1em}}}
\NewDocumentCommand{\rforall}{O{}}{{\hskip 1em\rlap{$\forall_{\IfValueTF{#1}{#1}{}}$}}}

\NewDocumentCommand{\TXDist}{}{\Distribution_{\tilde{x}}}

\NewDocumentCommand{\TheoremName}{m}{\emph{(#1)}}

\AtAppendix{\counterwithin{thm}{section}}
\NewDocumentCommand{\NewTheorem}{smm}{%
    \IfBooleanTF{#1}%
        {\newtheorem*{#2}{#3}}%
        {\newtheorem{#2}[thm]{#3}}%
    }
\NewTheorem{theorem}{Theorem}
\NewTheorem{lemma}{Lemma}
\NewTheorem{corollary}{Corollary}
\NewTheorem{fact}{Fact}
\NewTheorem{remark}{Remark}

\makeatletter
\NewTheorem{thm@definition}{Definition}
\NewDocumentEnvironment{definition}{m}{\begin{thm@definition}\TheoremName{#1}}{\end{thm@definition}}
\makeatother

\NewDocumentEnvironment{comment}{+b}{}{}

\begin{document}

\title{An access model for quantum encoded data}
\author{Miguel Mur\c{c}a}
\orcid{https://orcid.org/0000-0003-0651-7847}
\email{miguel.murca@tecnico.ulisboa.pt}
\affiliation{Physics of Information and Quantum Technologies group\char`,{} Centro de F\'{i}sica e Engenharia de Materiais Avan\c{c}ados (CeFEMA)\char`,{} Lisbon\char`,{} Portugal}
\affiliation{Instituto Superior T\'{e}cnico\char`,{} Universidade de Lisboa\char`,{} Lisbon\char`,{} Portugal}
\affiliation{PQI --- Portuguese Quantum Institute\char`,{} Lisbon\char`,{} Portugal}

\author{Paul K.\ Faehrmann}
\orcid{https://orcid.org/0000-0002-8706-1732}
\affiliation{Dahlem Center for Complex Quantum Systems\char`,{} Freie Universit\"{a}t Berlin\char`,{} 14195 Berlin\char`,{} Germany}

\author{Yasser Omar}
\orcid{https://orcid.org/0000-0003-2333-5762}
\affiliation{Physics of Information and Quantum Technologies group\char`,{} Centro de F\'{i}sica e Engenharia de Materiais Avan\c{c}ados (CeFEMA)\char`,{} Lisbon\char`,{} Portugal}
\affiliation{Instituto Superior T\'{e}cnico\char`,{} Universidade de Lisboa\char`,{} Lisbon\char`,{} Portugal}
\affiliation{PQI --- Portuguese Quantum Institute\char`,{} Lisbon\char`,{} Portugal}

\begin{abstract}
    We introduce and investigate a data access model (\kern-1pt\textit{approximate sample and
    query}) that is satisfiable by the preparation and measurement of block encoded
    states, as well as in contexts such as classical quantum circuit simulation or Pauli
    sampling. We illustrate that this abstraction is compositional and has some
    computational power. We then apply these results to obtain polynomial
    improvements over the state of the art in the sample and computational complexity of
    distributed inner product estimation. By doing so, we provide a new interpretation for
    why Pauli sampling is useful for this task. Our results partially characterize the
    power of time-limited fault-tolerant quantum circuits aided by classical computation.
    They are a first step towards extending the classical data Quantum Singular Value Transform
    dequantization results to a quantum setting.
\end{abstract}

\maketitle

\section{Introduction}

Quantum computation subsumes classical computation. This is well-known, but it is stated
under the assumption that the quantum computer is operating fault-tolerantly, i.e., that
noise in the computation can be effectively suppressed. Despite ongoing and advancing
efforts towards reaching this regime \cite{GoogleAI2023, DeCross2024, Gidney2024,
Brock2024, Acharya2024, Wills2024, Zhou2024}, to do so seems practically challenging.
Thus, the current reality of quantum and classical computation is at odds with the first
statement. This motivated a pragmatic exploration of the power of the existing limited
computers \cite{Kandala2017, Cerezo2022, Bharti2022}, typically designated Noisy
Intermediate-Scale Quantum (NISQ) devices, seeking to use them to go beyond classically
tractable problems (or ``achieve quantum advantage''). Careful analysis of these
approaches revealed subtle caveats on the power of NISQ computers and algorithms: they
must contend with optimization issues \cite{McClean2018, Larocca2024}, ingenious
classical simulation techniques that exploit the presence of noise \cite{Hner2017,
Huang2020Sim, Markov2020, Tomislav2023}, and questions regarding the role of loading in
and reading out of data, and whether that data is intrinsically quantum or not
\cite{Huang2020Shadow, Huang2021, Huang2022Quant, Cerezo2022, Chen2023}.

These considerations motivate a related, if more theoretical, question: even if adopting
a fault-tolerant quantum circuit setting, how powerful are hybrid algorithms: those that
conjugate limited quantum computation (for some notion of limited) with classical
computation? How strong can these limitations be before the model becomes at most as
powerful as fully classical computation? Answers to these questions have the potential to
impact the practical reality of NISQ: if a hybrid algorithm that does not have to contend
with noisy operations can be proven to be unable to achieve advantage in some task, then surely a
hybrid algorithm that furthermore must deal with noise cannot perform that same task. An
instance of this argument is found in Ref.~\cite{Chen2023}. Another example is found in
quantum query complexity results: a significant argument for the usefulness of quantum
computation is the proven separation between classical and quantum computation under an
oracle model \cite{Bernstein1997}. If this separation were to not hold under some
limitation on the number of queries performed coherently, this would argue against the
usefulness of NISQ algorithms and computation. (This is not the case:
Ref.~\cite{Aaronson2015} shows the existence of a problem for which a single quantum query
is sufficient for an exponential query quantum advantage, and Ref.~\cite{Sun2019} shows
that this separation also holds for total functions.)

\subsection{Prior art}
\label{subsec:prior-art}

The line of work started by Van den Nest \cite{vdNest2011} and developed by Tang, Chia
\et al.~\cite{Tang2019, Chia2020, LeGall2023} arguably addressed the power of quantum
computation from this perspective. Ref.~\cite{vdNest2011} showed that incorporating the
readout and classical post-processing into the access model of a quantum algorithm resulted in a model that could be
(efficiently) satisfied classically for some nontrivial cases. These included highly
entangled states, classes of states beyond Clifford \cite{Aaronson2004} and matchgate
\cite{Valiant2002} states, the Quantum Fourier Transform, and even quantum subroutines of
Simon's \cite{Simon1997} and Shor's \cite{Shor1994} algorithms. Therefore Ref.~\cite{vdNest2011} evidenced, early on, that
having access to limited fault-tolerant quantum computation with classical
post-processing \emph{could} be no better than purely classical processing: most
strikingly illustrated by the results on Shor's and Simon's algorithms, this foreshadowed
the relevance of data loading and readout in computational advantage.\footnotemark

\footnotetext{It also provided
further evidence that the power of quantum algorithms (with respect to classical
algorithms) was rooted beyond entanglement or an exponentially large representation
space, although this was already known from results such as those in
Refs.~\cite{Valiant2002, Aaronson2004}.}

Refs.~\cite{Tang2019, Chia2020} leveraged this line of thought (combining it with
the classical approach of Ref.~\cite{DKLR1995}) to address a complementary question:
generalizing the access model of Ref.~\cite{vdNest2011} to matrices, what was the computational power of this kind of
access with post-processing? Given that this access model could be taken to reflect
measurements on a quantum RAM, this complementary question arguably addressed the
power of a classical agent with access to limited fault-tolerant quantum computers.\footnotemark\ Their
results (where Ref.~\cite{Chia2020} strengthened and extended those of
Ref.~\cite{Tang2019}) were surprisingly strong: such access to classical data, together with
classical post-processing, could be enough to simulate a Quantum Singular
Value Transformation (QSVT) \cite{Gilyn2019} on the data, for sufficiently well behaved (for example, of low rank)
data. Note that the QSVT framework is known to encompass a
large number of significant quantum algorithms, including amplitude estimation
\cite{Brassard2002} and matrix inversion \cite{Harrow2009}, justifying the strength of
this statement.

\footnotetext{It is worth noting that the authors of Ref.~\cite{Chia2020} take a different
    point of view: given that there is not a known efficient way to prepare
    a quantum RAM, the authors interpret their results as evidence that a classical
    agent has no use for quantum power in the studied setting. This is because if the
    agent is willing to be inefficient in preparing a quantum RAM with their data, then
    there are classical (inefficient) procedures that suffice to satisfy the access model
    proposed by the authors.}

To better understand our contributions, it is useful to discuss in more detail how Chia
\et al.~\cite{Chia2020} obtained this result. Informally speaking, these authors showed
the dequantization of QSVT on classical data by first abstracting the power of
a quantum RAM (plus classical access to the loaded data) into a small set of promises on
how this data could be accessed: the \textit{sample and query} (\SQ) access model.
For a vector, these promises were that:
\begin{definition}{\textit{Sample and query} access to vectors --- Informal statement of Def.~\ref{def:sq}}
    \begin{enumerate}[nosep]
        \item One may sample an index by the corresponding entry's $\ell_2$ weight (suggesting
            a quantum measurement), taking time $\Sample$.
        \item One may read exactly the $\ell_2$ norm of the vector, taking time $\QNorm$.
        \item One may read exactly any given entry of the vector, taking time $\Query$.
    \end{enumerate}
\end{definition}
Chia \et al.\ then showed that \SQ\space access \textit{composed}, in the sense that if one had
\SQ\space access to vectors $x$ and $y$, one could compose this into \SQ\space access to a linear
combination of $x$ and $y$; or, extending the \SQ\space access model to matrices, that one
could compose \SQ\space access to $X$ and $Y$ into \SQ\space access to the inner or outer product,
or to the Kronecker product, of $X$ and $Y$. This sufficed to show that one could start
with \SQ\space access to, for example, matrix $A$ and vector $x$, and ``build up'' to a QSVT
$f(A^\dagger A) x$. It is worth noting that, to show composability, the authors had to also show
that \SQ\space \textit{computed}, in the sense that, for example, having \SQ\space access to
$x$ and $y$ sufficed to approximate the inner product of $x$ and $y$ (a perspective more
prevalent in Ref.~\cite{Tang2019}). 

More rigorously, one finds that speaking only of the \textit{sample and query} model is
insufficient to make the above argument, and rather that the notion of \SQ\space access needs to be
relaxed to a more general \textit{$\phi$-oversample and query} ($\SQ[\phi]$) access model. (It
suffices to think of $\phi$ as some relaxation parameter, where the greater $\phi$ is
than $1$, the more imperfect the access is.) Under this relaxation, Ref.~\cite{Chia2020}
showed that $\SQ[\phi]$ is approximately closed under arithmetic operations.

Since the SQ model was defined to abstract quantum RAM access to data, one might expect
that the ability to prepare and measure a given state (possibly performing some auxiliary
computation before measurement) would suffice to satisfy \SQ\space access to the state's
representation in the computational basis (for example). But the consequences of this
would be remarkable: because the results discussed above rely only on the data access promises of
\textit{sample and query}, all of the results of Ref.~\cite{Chia2020} would immediately apply to the scenario of
preparation and measurement. The composability of the \textit{sample and query} model would
immediately imply that something akin to a linear combination of unitaries
\cite{Childs2012} could be attained at a known overhead, without having to prepare the
various states coherently; and the dequantization results would imply that QSVTs could be
carried out classically for sufficiently well conditioned states. Indeed, as noted by the
authors of Ref.~\cite{Chia2020} themselves, and later in Ref.~\cite{Colter2021}, the
\textit{sample and query} access model is \emph{not} compatible with the case
where one is given the ability to prepare and measure a general quantum state; in other words,
when the input data is not classical. Thus, while applying to the particular case of
quantum RAM with classical computation, the results of \textit{sample and query} access
do not immediately apply to all limited fault-tolerant quantum computation aided by classical
post-processing.

\begin{figure*}
    \begin{subfigure}[b]{.5\linewidth}
        \includegraphics[width=\linewidth]{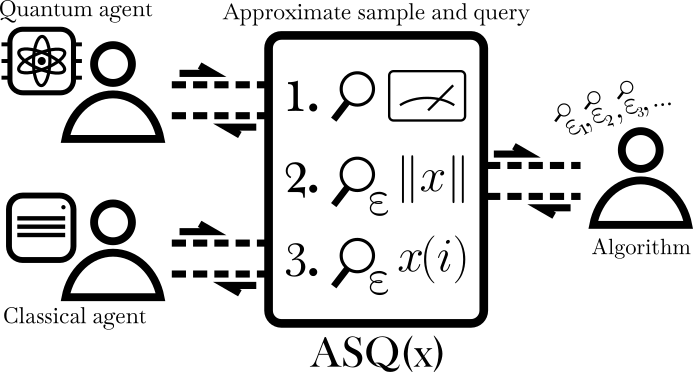}
        \caption{The \textit{approximate sample and query} model.} \label{subfig:asq}
    \end{subfigure}%
    \begin{subfigure}[b]{.5\linewidth}
        \centering
        \includegraphics[width=.83\linewidth]{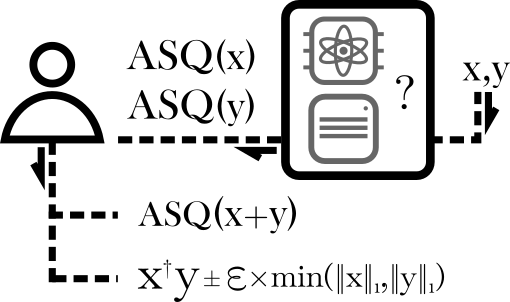}
        \vskip 3.5pt
        \caption{The power of the \textit{approximate sample and query} model.} \label{subfig:power}
    \end{subfigure}
    \caption{The \textit{approximate sample and query} (\ASQ) access model restricts access to a
        vector to three rules (Fig.~\ref{subfig:asq}). These can be satisfied by
        quantum or classical agents in the right conditions. By developing algorithms that
        operate over this abstraction, their runtime can be described in terms of the time
        necessary to carry out each access operation. If a classical agent can satisfy
        each operation fast enough, a dequantization result follows. In this paper, we
        establish some of the power of \ASQ{} (Fig.~\ref{subfig:power}): we show that an
        agent can \emph{compose} \ASQ\space into linear combinations, and that they can use
        \ASQ\space access to \emph{compute} the inner product.}
    \label{fig:summary}
\end{figure*}

\subsection{Our contributions}

In this work, we adapt the \textit{sample and query} access model, such that it \emph{can} be
satisfied by preparation and measurement of a quantum state (with some auxiliary
computation before measurement). We call this modified access model \textit{approximate sample
and query} (or \ASQ). This modified model addresses the following points, which the
original \SQ\space model cannot accommodate:
\begin{enumerate}[nosep]
    \item Computing amplitudes or norms of quantum states to exponential precision by
        preparation and measurement is expected to be inefficient \cite{Harrow2009};
    \item Estimating a certain amplitude to varying precision results in varying runtime
        (and, conversely, a given task may only require a certain degree of precision);
    \item Quantum operations may succeed only probabilistically, and this success may not
        necessarily be signaled.
\end{enumerate}
The \textit{approximate sample and query} model addresses these points by incorporating
the notion of precision (and its relationship with runtime), and admitting that the
promises are satisfied only with some probability. Thus, informally, the promises of
\ASQ\space for a certain vector $x$ become:
\begin{definition}{\kern-.5pt\textit{Approximate sample and query} for a vector $x$ --- Informal statement of Def.~\ref{def:asq}}
\begin{enumerate}[nosep]
    \item One may sample an index of $x$ by the corresponding entry's $\ell_2$ weight (suggesting a quantum measurement), or
        fail with probability at most $1/3$ by raising an error flag. This action is
        assumed to take some time $\Sample$.
    \item One may estimate the $\ell_2$ norm of $x$ to some absolute error $\epsilon$, or
        fail undetectably with probability at most $1/3$. This action takes time that
        varies with the choice of $\epsilon$, and so is denoted $\QNorm(\epsilon)$.
    \item One may estimate a certain entry of $x$ to some absolute error $\epsilon$, or
        fail undetectably with probability at most $1/3$. This action also takes time
        varying with the choice of $\epsilon$, thus denoted $\Query(\epsilon)$.
\end{enumerate}
\end{definition}

Note that the model considers different admissible modes of failure for each operation.
This is because if one can produce an estimation of a value to error $\epsilon$ with some
probability bounded away from $1/2$, then one may show (by a Hoeffding bound) that a
logarithmic overhead suffices to ``boost'' the correctness probability.\footnotemark\ However, the same
cannot be done for the index sampling procedure, and so, in that case, undetectable
failures become problematic.

\footnotetext{This is achieved by taking the median of several estimates logarithmic
on the desired failure probability, a procedure typically referred to as ``median of
means''.}

We find that, nonetheless, these demands on how data may be accessed nicely conform to
access to a quantum state: sampling may fail if, for example, post-selection for the
sampled state fails. But, in that case, the failure is evident. Contrarily, the estimation of
an amplitude by, for example, quantum phase estimation returns a value that is correct
with high probability, but without any indication of success. Leveraging existing results
on quantum tomography \cite{vanApeldoorn2023}, we characterize the runtimes $\Sample$,
$\QNorm(\epsilon)$, and $\Query(\epsilon)$ for the case where one may prepare quantum
states that encode a (subnormalized) vector in a given subspace of their computational
basis representation, without performing much more quantum computation before measuring.
[In particular, for an $\QubitsN$ qubits state, we show how to satisfy \ASQ\space access to the
subnormalized vector with $\BigO(\QubitsN)$ quantum gates over $\BigO(\log\QubitsN)$ depth
between the state preparation and measurement.] This is given in
Remark~\ref{remark:asq-prepmeasure}. We show that these runtimes do not scale, for
example, more than logarithmically with the dimension of the vector.

Perhaps surprisingly, we show that \ASQ\space access can also be satisfied in other
contexts. In particular, we characterize the runtimes $\Sample$, $\QNorm(\epsilon)$, and
$\Query(\epsilon)$ for two setups other than preparation and measurement. First, we find
that the promises of \ASQ\space may be satisfied by a classical agent for a given state $x$ if
they are given a description of $x$ by means of a of a quantum circuit that prepares it,
and that circuit has few $T$ gates and otherwise uses Clifford gates. This leverages
existing results on the classical simulation of such circuits \cite{Bravyi2016}, and is
given in Remark~\ref{remark:asq-simt}. Secondly, we find that if an agent has the ability
to perform \textit{Pauli sampling} of a certain quantum state $x$, as well as the ability
to prepare and measure $x$, then one can satisfy \ASQ\space access to the \emph{representation
of $x$ in a Pauli basis}. (Pauli sampling here means sampling a Pauli string with
probability proportional to its expectation value.) This is given in
Remark~\ref{remark:asq-pauli}.

We summarize these results in Fig.~\ref{subfig:asq}.

Our motivation for investigating \textit{approximate sample and query} is threefold.
First: by defining a model that describes (to some extent) the ability to measure a
quantum state, results on the computational power of this model (plus classical
computation) establish results on the computational power of time-limited fault-tolerant
quantum computation, when aided by classical computation. For example: a result on the
complexity of computing the inner product of $x$ and $y$ given \ASQ\space access to $x$ and $y$
establishes an upper bound to the computational complexity of the distributed computation
of the inner product of two states. This is because if Alice and Bob can prepare, locally
and respectively, states $\ket{\psi}$ and $\ket{\phi}$, then Alice and Bob have \ASQ\
access to the representation of $\ket{\psi}$ and $\ket{\phi}$ in the computational basis
(as stated above). Thus, these can be taken to be $x$ and $y$, from whence using the \ASQ\
result, Alice and Bob can compute the inner product between the two states with local
operations and classical communication between them. Secondly: the compositional aspect of
\SQ\space is conceptually attractive, and we wish to find a similar property for the case of
quantum data. As mentioned in the previous section, the power to prepare a linear
combination of states is a significant quantum operation \cite{Childs2012}. If one admits
that the \ASQ\space model captures (to some extent, and as established with computational
results) the power of being able to prepare and measure a state, then a composition result
on linear combinations shows that it is possible to (and effectively characterizes the
overhead of) ``faking'' access to a linear combination of states, when given access to
each state individually and complementing this with classical computation. Thirdly: given
that Chia \et al.~\cite{Chia2020} could dequantize QSVT based on the language of
\textit{sample and query}, this raises the question of whether an analogous result can be
obtained for a sufficiently similar model that nonetheless reflects quantum access. Beyond
the evident practical significance of such a quantization, as argued in the previous
section, such an analogous result would provide much insight into the question: ``in
what conditions can a quantum computation be halted, and the remaining computation be
carried out classically, without too much overhead?''

We address the first two points directly, and, with this, take a first step in addressing
the third point. Namely, we find that the modifications of \ASQ\space (with respect to \SQ\space and
as discussed above) mean the results of Chia~\et al.\ no longer apply. Thus, we first
show that the \ASQ\space model remains compositional, not going as far as showing that it is
closed under arithmetic, but rather taking a first step in this direction by focusing on
composition by linear combinations (motivated by linear combinations of unitaries, as
stated). This requires us to define a relaxed access model, $\ASQ[\phi]$ (in line with
$\AS[\phi]$), after which we show the following:
\begin{lemma} \TheoremName{Linear combinations of \ASQ\space --- Informal statement of Lemmas~\ref{lemma:lin-comb}~and~\ref{lemma:lin-comb-rand}.}
    Given $\ASQ[\phi]$ access to complex vectors $x_1, \ldots, x_\tau$, and complex coefficients
    $\lambda_1, \ldots, \lambda_\tau$, one may obtain $\ASQ[\phi']$ access (with certainty)
    to the linear combination $\lambda_1 x_1 + \cdots + \lambda_\tau x_\tau$, for suitable
    $\phi$. If the vectors are approximately orthogonal and of constant and same size,
    then $\phi' \sim \phi\tau^2$. If, instead, one is willing to fail the construction with
    some probability, then $\phi' \sim \phi\tau$, matching that of Chia~\et al.'s Lemma~3.6.
\end{lemma}

We also show that the \ASQ\space model still holds some computational power. Again the
results of Chia~\et al. no longer apply, so we focus on the task of inner product
estimation. We show the following:
\begin{lemma} \TheoremName{Inner product estimation with \ASQ\space --- Informal statement of Lemmas~\ref{lemma:inprod-asym}~and~\ref{lemma:inprod-sym}.}
    \label{lemma:inform-inprod}
    Suppose (without loss of generality, normalized) complex vectors $x$ and $y$ of
    exponentially large dimension. Given efficient $\ASQ[\phi]$ access to $x$ and the ability to
    compute entries of $y$, one may compute $x^\dagger y$ to error $\epsilon \norm{y}_1$
    efficiently with a probability of at least $2/3$. If one is given efficient $\ASQ[\phi]$ access to
    \emph{both} $x$ and $y$, then $x^\dagger y$ is efficiently computable to error
    $\epsilon \min(\norm{x}_1, \norm{y}_1)$ with probability at least $2/3$.
\end{lemma}

The reader may be alarmed by the presence of the $1$-norm: it may be of the order of the
dimension of the vectors, which is taken to be exponentially large. However, we later find
this dependence to be sufficient. We furthermore interpret this quantity as a concrete
measure of ``peakedness'', which may be controlled by a suitable choice of representation
basis.

We summarize these results in Fig.~\ref{subfig:power}.

Finally, we show that the \ASQ\space model has explanatory power and that the promises encoded
by this model suffice to obtain improved results on the topic of distributed inner product
estimation. We show that, under the light of the \ASQ\space model, Pauli sampling is useful to
compute the distributed inner product \cite{Anshu2022,Hinsche2024} because, for
sufficiently well-conditioned states, their representation in the Pauli basis has
reasonably small $1$-norm. We do so by first specializing Lemma~\ref{lemma:inform-inprod}:
\begin{lemma} \TheoremName{Inner product estimation with exactly \ASQ\space for real vectors --- Informal statement of Lemma~\ref{lemma:inprod-real}.}
    If $x$ and $y$ are (without loss of generality, normalized) \emph{real} vectors, and
    one has \ASQ\space access to them (\emph{not} $\ASQ[\phi]$ access), and if $\kappa$ is
    an upper bound to $\max(\norm{x}_1, \norm{y}_1)$, then the dependence on $\kappa$ of
    Lemma~\ref{lemma:inform-inprod} may be improved. This requires that the runtime of
    estimating an entry [$\Query(\epsilon)$] is not too large.
\end{lemma}

From here, we combine three observations: the first is that the inner product of two
states may be computed in a Pauli basis (where the representations are real). The second
is that the $1$-norm of a state's representation in the Pauli basis is directly connected
to certain measures of ``stabilizerness'' \cite{Campbell2011,Leone2022}. The third
observation is that prior art on Pauli sampling \cite{Hinsche2024} gives explicit
procedures and runtimes for Pauli sampling, and thus gives $\Sample$, $\QNorm(\epsilon)$
and $\Query(\epsilon)$ for \ASQ\space access to the states' representation in the Pauli basis.
By adapting this prior art somewhat (Corollary~\ref{corollary:pauli-samp}) and directly
computing overalps between Pauli amplitudes, we establish
\ASQ\space runtimes for Pauli sampling under a constraint in the ``stabilizerness'' measure
corresponding to the $1$-norm of the Pauli basis representation of the states.

This allows us to state a distributed inner product estimation algorithm that polynomially
improves over the runtime and sample complexity of the state of the art
(Theorem~\ref{thm:pauli-inprod}).

\section{Notation and preliminaries}

\subsection{Notation}

For some integer $n$, $[n]$ denotes the set $\{1, 2, \ldots, n\}$. $\I$ denotes the
imaginary unit. We denote vectors by lowercase letters. The $i$th entry of some vector
$x$ is denoted $x(i)$. $\norm{\cdot}_p$ refers to the $p$-norm for vectors, and to the
induced $p$-norm for matrices (thus, for some matrix $A$, $\norm{A}_2$ refers to the
induced two norm, or the Schatten $\infty$-norm). $\norm{A}_F$ refers to the Frobenius
norm (or the Schatten $2$-norm) of $A$. When unspecified, $\norm{\cdot}$ refers to
$\norm{\cdot}_2$. We denote probability distributions by upright lowercase letters, e.g.,
$\Distribution$. We denote the probability of a certain event (for example, $X=x$) as
$\Prob[X=x]$. Given a probability distribution $\Distribution$, we denote that a
random variable $X$ follows that distribution as $X \sim \Distribution$. We likewise
denote the probability of observing value $x$ under this distribution as
$\Distribution(x)$ (i.e., this is equivalent to denoting $\Prob[X=x]$ for $X \sim
\Distribution$). We will conflate the random variable with the value it takes where it is
clear from context. The expectation of $X$ is denoted $\Expectation X$, and its variance
$\Var X$. For an arbitrary vector $x$ of dimension $d$, we often refer to a corresponding
distribution, denoted $\Distribution_x$, to be a probability distribution with support
$[d]$, satisfying ${\Distribution_x}(i) \equiv \abs{x(i)}^2/\norm{x}^2$ (i.e.,
$\Distribution_x$ is the $\ell_2$ sampling distribution of $x$). $\Identity_d$ denotes
the $d$-dimensional identity matrix. Approximate quantities or estimates are generally
denoted with a hat. All logarithms ($\log$) are natural unless otherwise noted.
$\Poly(x)$ denotes some polynomial over $x$, and likewise $\Polylog(x)$ denotes some
polynomial over $\log x$. $\log^k x$ denotes the $k$th power of $\log x$. The Fraktur font is used to make clear
when referring to a number of qubits (thus, the base-$2$ logarithm of the dimension of
the underlying state vector). E.g., an $\QubitsN$ qubits pure
quantum state has a $d \equiv 2^\QubitsN$-dimension underlying quantum state vector. ``$x$
is defined as $y$'' is denoted $x \Defined y$, and ``$x$ taken to be the definition of $y$''
is denoted $x \RDefined y$. Moreover, $x \equiv y$ emphasizes a coincidence by
the definitions of $x$ and $y$. ``Big O notation'' [$\BigO(\cdot)$] is used with its
usual definition, with the use of ``big omega'' [$\Omega(\cdot)$] and ``big theta''
[$\Theta(\cdot)$] as popularized by Knuth \cite{Knuth1976}. The notation
$\tilde\BigO(\cdot)$ indicates that logarithmic terms on the quantities involved in the
expression are suppressed. E.g., $\tilde\BigO(xy + z)$ indicates terms of $\Polylog(x)$,
$\Polylog(y)$, and $\Polylog(z)$ are suppressed in the argument.

\subsection{Block encodings and projective encodings}
\label{sec:block-enc}

Quantum circuits encode, forcibly, unitary maps. However, many mappings of interest are
non-norm preserving. To deal with this, as well as to more generally encode non-unit norm
vectors and matrices, one possible strategy is to encode the object or map in a suitable
subspace, i.e., perform a block encoding:
\begin{definition}{Block encoding \cite{Gilyn2019}}
    For $A$ an $\QubitsN$ qubit operator, if $\alpha,\epsilon > 0$ and $\QubitsM \in
        \Natural$, then an $(\QubitsN + \QubitsM)$-qubits unitary $U$ is an $(\alpha,
        \QubitsM, \epsilon)$-block encoding of $A$ if
    \begin{equation}
        \norm{A - \alpha ( \bra{0}^{\otimes\QubitsM} \otimes \Identity_{2^{\QubitsN}}) U (\ket{0}^{\otimes\QubitsM} \otimes \Identity_{2^{\QubitsN}}) } \leq \epsilon.
    \end{equation}
\end{definition}

Informally, $U$ is a block encoding of $A$ because the block of $U$ spanned by the
$\ket{0}^{\otimes\QubitsM}$ subspace of the first $\QubitsM$ qubits is approximately
proportional to $A$. This scaling is necessary because a block of $U$ can never have norm
greater than $U$ itself, which cannot have norm greater than $1$. Therefore, it
is only possible to block encode $A$ if $\alpha \geq \norm{A}$.

It is possible to define an analogous notion for vectors:
\begin{definition}{Block encoding --- vectors}
    For $v$ a $2^\QubitsN$ dimensional complex vector, defined up to a global phase, and
    if $\alpha, \epsilon > 0$ and $\QubitsM \in \Natural$, then an $(\QubitsN +
        \QubitsM)$-qubits quantum state $\ket{\psi}$ is an $(\alpha, \QubitsM,
        \epsilon)$-block encoding of $v$ if
    \begin{equation}
        \norm{v - \alpha \, (\bra{0}^{\otimes\QubitsM} \otimes \Identity_{2^\QubitsN}) \ket{\psi}} \leq \epsilon.
    \end{equation}

    An $\QubitsM$-block encoding of $v$ refers to a $(1, \QubitsM, 0)$-block encoding of
    $v$. Likewise an $(\QubitsM, \epsilon)$-block encoding of $v$ refers to a $(1,
    \QubitsM, \epsilon)$-block encoding of $v$.
\end{definition}

Block encodings may be further generalized to the notion of ``projected encodings'',
whereby there is a more generic projection operator, such that the operator or vector of
interest is defined in the projector's image:
\begin{definition}{Projected encoding \cite{Gilyn2019}}
    Let $\Pi$ and $\tilde{\Pi}$ be projection operators, and $U$ a unitary operator. Then
    $(U, \Pi, \tilde{\Pi})$ form a projected encoding of an operator $A$ if $\tilde{\Pi} U
        \Pi \equiv A$.
\end{definition}
This is straightforwardly generalized to approximations and operators of norm greater
than one.

\subsection{One-sided and two-sided error}
\label{sec:sided-err}

A probabilistic algorithm (as are usually quantum algorithms) may fail in multiple,
operationally different ways. For this reason, we wish to differentiate these failure modes in the analysis below. Borrowing from the language of probabilistic Turing
machines and decision problems \cite{Arora2009}, we will refer to ``one-sided'' and
``two-sided'' types of error:
\begin{definition}{One-sided error}
    A probabilistic decision algorithm fails with one-sided error $p$ if with probability
    $1-p$ the algorithm correctly accepts or rejects the input and with probability $p$
    rejects valid input. An algorithm with one-sided error never accepts invalid input.
\end{definition}

For a non-decision algorithm, suppose that the algorithm outputs the result and a flag
indicating if the result is valid. A one-sided algorithm may falsely report that the
provided output is wrong, but will never report bogus output as correct. For this
reason, we refer to probabilistic algorithms that may detectably fail with some
probability as ``one-sided''. For example: state post-selection in block encodings
(cf.~section~\ref{sec:block-enc}).

\begin{definition}{Two-sided error}
    A probabilistic decision algorithm fails with two-sided error $p$ if with probability
    $1-p$ the algorithm correctly accepts or rejects the input and with probability $p$
    randomly accepts or rejects the input.
\end{definition}

Again we extend this definition to non-decision algorithms in the same fashion: if a
probabilistic algorithm fails in an undetectable way, i.e., returns bogus output with
some error probability, then we refer to them as ``two-sided''. For example: estimating a
coin's bias from sampling, since one may be particularly unlucky and draw only heads,
regardless of the coin's weight.

\subsection{Sample and Query}
\label{sec:sq}

As discussed in Section~\ref{subsec:prior-art}, the \textit{sample and query} access model
is a data access model defined to investigate the role of a quantum RAM in obtaining
quantum advantage with classical data. The \textit{sample and query} (\SQ) access model is
defined as follows:
\begin{definition}{Sample and query access \cite{vdNest2011,Chia2020}} \label{def:sq}
    For a vector $v \in \Complex^d$, one has \textit{sample and query} access to $v$ if one can:
    \begin{enumerate}[label*={SQ\arabic*.}, ref={SQ\arabic*}, align=left, labelindent=3pt, leftmargin=!, labelwidth=!]
        \item read, for fixed $i\in[d]$, $v(i)$ to arbitrary precision, in time $\Query$;
              \label{def:sq-query}
        \item obtain independent samples $i\in[d]$, following a distribution with
              probability mass function $\Distribution(i) = \abs{v(i)}^2/\norm{v}^2$, in
              time $\Sample$; \label{def:sq-sample}
        \item read $\norm{v}$ to arbitrary precision, in time $\QNorm$.
              \label{def:sq-norm}
    \end{enumerate}
\end{definition}

In the context where $v$ is known classically and can be loaded into a quantum RAM (i.e.,
if one may prepare the state $\sum_{i\in[d]} [v(i)/\norm{v}] \ket{i}$), it is clear that
the requirements of \SQ\space access are satisfied. Furthermore, because every entry can be
known to arbitrary precision (or, at least, to exponential precision in runtime
\cite{vdNest2011}), the probability of any $i\in[d]$ can be known exactly (or to
exponential precision) in time $\BigO(\Query + \QNorm)$. These requirements may be split
into query access (\ref{def:sq-query}) and sample access (\ref{def:sq-sample} and
\ref{def:sq-norm}).

Ref.~\cite{Chia2020} introduces also the notion of $\phi$-oversampling. Say that, for $d
\in \Natural$, $\Distribution$ and $\AltDistribution$ are distributions supported on
$[d]$. For $\phi \geq 1$, we say that $\Distribution$ $\phi$-oversamples
$\AltDistribution$ if, for all $i \in [d]$,
\begin{equation}
    \Distribution(i) \geq \frac{1}{\phi} \AltDistribution(i).
\end{equation}
If, for some $i \in [d]$, $\Distribution(i)$, $\AltDistribution(i)$ and $\phi$ are known
exactly, then a sample of $\Distribution$ can be transformed into a sample of
$\AltDistribution$ by rejection sampling. This motivates the definition of \textit{$\phi$-oversample and
query access} ($\SQ[\phi]$) to a vector:

\begin{definition}{\textit{$\phi$-Oversample and query} access} \label{def:osq}
    For a vector $v \in \Complex^d$ and $\phi \geq 1$, one has $\phi$-oversample and query
    access to $v$ if one has query access to $v$ and sample and query access to a vector
    $\tilde{v} \in \Complex^d$, satisfying:
    \begin{enumerate}
        \item for all $i \in [d]$, $\abs{\tilde{v}(i)}^2 \geq \abs{v(i)}^2$;
        \item $\norm{\tilde{v}}^2 \leq \phi \norm{v}^2$.
    \end{enumerate}
\end{definition}
Note that, with these conditions, $\Distribution_v$ is $\phi$-oversampled by
$\Distribution_{\tilde{v}}$.

Ref.~\cite{Chia2020} generalizes the \SQ\space model to matrices. Denoting by $A(i,*)$ the
$i$th row of $A$:
\begin{definition}{\textit{$\phi$-Oversample and query access} --- matrices \cite{Chia2020}}
    For a matrix $A \in \Complex^{m \times n}$, one has \SQ\space access to $A$ if one has \SQ\
    access to every row of $A$, and \SQ\space access to the vector of row norms, i.e., \SQ\
    access to the vector $a \in \Real^m$, defined by its entries as $a(i) \Defined
        \norm{A(i, *)}$. One has $\ASQ[\phi]$ access to $A$ if one has query access to every
    row of $A$ and $\SQ[\phi]$ to $\tilde{A} \in \Complex^{m \times n}$, satisfying
    \begin{enumerate}
        \item for all $(i,j) \in [m]\times[n]$, $\abs{\smash{\tilde{A}(i,j)}}^2 \geq
              \abs{A(i,j)}^2$;
        \item $\norm{\smash{\tilde{A}}}_F^2 \leq \phi \norm{A}_F^2$.
    \end{enumerate}
\end{definition}

The model of \textit{oversample and query} is approximately closed under arithmetic operations
\cite{Chia2020}. The inclusion of \textit{oversampling and querying} --- rather than just \textit{sampling
and querying} --- is important to ensure this property. We refer to Ref.~\cite{Chia2020}
for details, but note that these closures include mapping from $\SQ[\phi]$ access to
matrices into $\SQ[\phi']$ access to approximately linear combinations thereof,
approximately the inner product thereof, or approximately general Lipchitz functions
(and, in particular, polynomials) thereof.

\subsection{Singly coherent copy tomography}
\label{sec:single-copy-tom}

In Ref.~\cite{vanApeldoorn2023}, van Apeldoorn \et al.\ investigate the feasibility of
tomography of a pure quantum state (say, $\ket{\psi}$) given the ability to prepare and
measure copies of $\ket{\psi}$, without computing more than one copy of $\ket{\psi}$ in
coherence. Contrary to what one might expect, it is possible to obtain an estimate of
$\ket{\psi}$'s state vector in the computational basis (up to a global phase) to
$\epsilon$ error in $\infty$-norm distance (i.e., up to $\epsilon$ error in each entry of
the tomography) in an amount of time and space that do not depend on the dimension of the
state. We restate here Ref.~\cite{vanApeldoorn2023}'s results significant to this work. We
rephrase the results to include the runtime requirements for each task (since the original
statements are concerned only with the necessary number of state preparations and measurements,
but provide explicit constructions).

\begin{theorem} \TheoremName{Estimation of state vector absolute values \cite[Proposition 13]{vanApeldoorn2023}}
    \label{thm:smpl-estim}
    Let $\epsilon > 0$, and $\delta \in [1/2, 1)$. Let $\ket{\psi} = \sum_{j \in [d]}
        \alpha_j \ket{j}$ be a quantum state, with $\alpha_j\in\Complex$, and let
    $\Distribution$, defined by its probability mass function as $\Distribution(j) =
        \abs{\alpha_j}^2$, $j \in [d]$, be the probability distribution of the outcomes of a
    computational basis measurement. Let $T$ be the time necessary to prepare and measure
    $\ket{\psi}$. Then, $\BigO[\log(d/\delta)\epsilon^{-2}T]$ time [and in particular
            $8\epsilon^{-2}\log(2d/\delta)$ measurements] are sufficient to learn the absolute
    values of $\alpha$ up to $\epsilon$ $\infty$-norm error, with success probability at
    least $1-\delta$, and using only $T$\negthinspace\space quantum computation time between measurements.
    These values are stored in $\BigO[\log(d/\delta)\epsilon^{-2}]$ space, and each entry
    can be read in $\BigO(\log d)$ time.
\end{theorem}

The above Theorem does not provide phase information on the entries of the state vector.
On the other hand, a global phase on the state is physically irrelevant. If one admits a
procedure performing tomography of the state up to a global phase, then just the ability
to prepare the state suffices to perform such a tomography (up to $\infty$-norm error).

\begin{theorem} \TheoremName{Estimation of state vector with phases by prepare and measure \cite[Proposition 19]{vanApeldoorn2023}}
    \label{thm:vad-tomography}
    Let $\epsilon > 0$ and $\delta \in [1/2, 1)$. Let $\ket{\psi}$ be a $d$-dimensional
    quantum state, which is prepared in black-box fashion in time $T$. Then,
    $\tilde\BigO[T \epsilon^{-4} \log^2(d/\delta)]$ time is sufficient to compute a sparse
    representation of the quantum state (in the computational basis), up to $\epsilon$
    $\infty$-norm error, with success probability at least $1-\delta$, and using at most
    $\BigO[T + \epsilon^{-2}\log(d/\delta)\log^3 d]$ quantum time between measurements.
    The quantum state is stored in $\BigO[\epsilon^{-2}\log(d/\delta)\log^2d]$ space, and
    each entry can be read in $\BigO(\log d)$ time.
\end{theorem}
See also Ref.~\cite{Shende2002} for the time necessary to perform some of the steps
outlined in the proof of \cite[Proposition 19]{vanApeldoorn2023}.

\subsection{Simulation of low T-gate count states}
\label{sec:low-t}

Circuits dominated by Clifford operations (Hadamard, phase, and controlled \textsc{not}
gates) admit efficient classical simulation, in the sense that one may approximately
sample from the distribution of measurements in the computational basis, and estimate the
probability of a selected output \cite{Bravyi2016}. In particular:
\begin{theorem} \TheoremName{Simulation of low $T$-gate count states \cite{Bravyi2016}}
    \label{thm:t-gate-sim}
    Suppose a state $\ket{\psi}$ of $\QubitsN$ qubits is (classically) given by a quantum
    circuit description of a unitary $U$, such that $U \ket{0}^{\otimes \QubitsN} =
    \ket{\psi}$. This description is taken to be of size polynomial in $\QubitsN$, and by
    the gates of the circuit, which are either Hadamard, phase, controlled \textsc{not},
    or $T$ quantum gates. Let $x$ be the representation of $\ket{\psi}$ in the
    computational basis. Then, if there are $t$ T-gates in this description, and $c$
    Clifford gates in this description, one may sample from a distribution
    $\AltDistribution$ supported on $[n]$, satisfying
    \begin{equation}
        \sum_{i \in [n]} \abs{\AltDistribution(i) - \Distribution_x(i)} \leq \Delta
    \end{equation}
    in time $\tilde\BigO[\QubitsN(\QubitsN + t)(c + t) + \QubitsN(\QubitsN + t)^3 + 2^t
    t^3 \QubitsN^3 \Delta^{-4}] = \tilde\BigO[\Poly(\QubitsN)(1 + 2^t \Delta^{-4})]$;
    and, for choice of $i \in [n]$, one may estimate $\Distribution_x(i)$ to error
    $\epsilon$ with failure probability $\delta$ in time $\BigO[(1+t)(c+t) + (\QubitsN +
    t)^3 + 2^t t^3 \epsilon^{-2} \log \delta^{-1}] = \BigO[\Poly(\QubitsN)(1 + 2^t
    \epsilon^{-2} \log \delta^{-1})]$.
\end{theorem}

\subsection{Pauli sampling and nonstabilizerness}
\label{sec:pauli-samp}

The $4^\QubitsN$-qubit phaseless Pauli string operators form a basis to the space of
$\Complex^{2^\QubitsN\times 2^\QubitsN}$. That is, if $X$, $Y$, and $Z$ denote the
corresponding single-qubit Pauli operators, then for an arbitrary\footnote{Without loss
of generality, but for concreteness, the ordering of the operators may be fixed in a
``tableau'' format, as introduced in \cite{Aaronson2004}. This method maps $(x_1, x_2,
\ldots, x_{2\QubitsN}) \in \{0,1\}^{2\QubitsN}$ to
$(X^{x_1}Z^{x_2})\otimes(X^{x_3}Z^{x_4})\otimes (\cdots) \otimes
(X^{x_{2\QubitsN-1}}Z^{x_{2\QubitsN}})$, modulo the resulting phases.} ordering of the
operators
\begin{equation}
    P_i \in \{I, X, Y, Z\}^{\otimes\QubitsN}\mathrlap{\, , \qquad i=1,2,\ldots,4^\QubitsN,}
\end{equation}
the set of operators $\{P_i\}_i$ forms a basis of $\Complex^{2^\QubitsN \times
2^\QubitsN}$. This implies that every density matrix $\rho$ of dimension $d$ admits a
decomposition of the form
\begin{equation} \label{eq:pauli-decomp}
    \rho \mkern7mu = \sum_{i \in [d^2]} \frac{\Tr(\rho P_i)}{\sqrt{d}} P_i
\end{equation}
where the factor of $d^{-1/2}$ is included to ensure orthonormality. Henceforth we will
denote the representation of $\rho$ in this basis as $\pi_\rho$, such that
\begin{equation} \label{eq:pauli-vec}
    \pi_\rho(i) = \frac{\Tr(\rho P_i)}{\sqrt{d}}\mathrlap{\, , \qquad i=1,2,\ldots,d^2.}
\end{equation}
Note that $\pi_\rho$ is a real-valued vector and that it has unit $2$-norm.

Pauli sampling corresponds to sampling $i \in [d^2]$ with probability $\pi_\rho(i)^2$
\cite{Leone2022}. While effectively performing this sampling can be challenging
\cite{Chen2022,Hinsche2024}, the ability to do so can be powerful, enabling, for example,
fidelity estimation in a constant number of samples \cite{Flammia2011}.

Furthermore, the representation of a pure state $\dyad{\psi}$ in a Pauli basis
($\pi_\psi$) relates closely to the ``nonstabilizerness'' of $\ket{\psi}$ (sometimes also
referred to as ``magic''). Recall that a state $\ket{\phi}$ is a stabilizer state (or
Clifford state) if there exists a unitary $U$ composed of controlled \textsc{not},
Hadamard and phase gates, such that $U\ket{0} = \ket{\psi}$. Equivalently, stabilizer
states of $\QubitsN$ qubits are those stabilized by a group of $2^\QubitsN$ Pauli
strings. Stabilizer states can be classically simulated and are less powerful than
classical circuits \cite{Aaronson2004, VandenNest2010}.  If a state has a large overlap
with few stabilizer states, then we say that it has low ``nonstabilizerness''.
This fact can be exploited to, for example, classically compute and simulate measurements of
the state \cite{Bravyi2016}. As shown in Ref.~\cite{Leone2022}, the $\alpha$-R\'{e}nyi
entropies of $\pi_\psi$ directly relate to various measures of ``nonstabilizerness'',
called stabilizer $\alpha$-R\'{e}nyi entropies, or simply $\alpha$-magic:
\begin{equation}
    M_\alpha(\ket{\psi}) = (1-\alpha)^{-1} \log[\sum_{i \in [d^2]} \pi_\psi(i)^{2\alpha}\mkern2mu] - \log d.
\end{equation}
If $H_\alpha(\cdot)$ denotes the $\alpha$-R\'{e}nyi entropy,
$M_\alpha(\ket{\psi}) \equiv H_\alpha(\Distribution_{\pi_\psi}) - \log d$. Furthermore,
from the properties of the R\'{e}nyi entropies, $M_\alpha(\ket{\psi}) \in [0, \log d]$,
with zero magic if and only if $\ket{\psi}$ is a stabilizer state. Also, this definition
unifies other definitions of nonstabilizerness in the literature. In particular we
highlight the connection made in Ref.~\cite{Leone2022} regarding ``stab-norm'':
\begin{equation}
    M_{1/2}(\ket{\psi}) \equiv 2 \log \StabNorm(\ket{\psi})
\end{equation}
where $\StabNorm(\ket{\psi}) \Defined d^{-1} \sum_{i\in[d^2]} \abs{\Tr(P_i \dyad{\psi})}$
is the ``stab-norm'' of $\ket{\psi}$ as defined in Ref.~\cite{Campbell2011}.

In Ref.~\cite{Hinsche2024}, Hinsche \et al.\ describe how to perform approximate Pauli
sampling on a pure state provided that its $M_0$ magic is sufficiently bounded. (Approximate Pauli sampling means that the sampling follows a distribution not too
different from the Pauli sampling distribution, in $\ell_1$ distance.) Their procedure
requires Bell measurements on $\rho\otimes\rho$, and has the following performance
guarantees:
\begin{theorem} \TheoremName{Performance guarantee for approximate Pauli sampling \cite[Theorem 8]{Hinsche2024}}
    \label{thm:hinsche-pauli}
    Let $\Delta \in (0,1)$, and $\delta \in (0,1)$. Let $\rho$ be a pure, $\QubitsN$
    qubit state with magic $M_0$, with no more than $\chi$ entanglement across any of its
    bipartitions. Then, there exists an algorithm sampling from a distribution
    $\AltDistribution$ such that, with overall probability $1-\delta$
    \begin{equation}
        \sum_{i\in[4^\QubitsN]} \abs{\AltDistribution(i) - \Distribution_{\pi_\rho}(i)} \leq \Delta,
    \end{equation}
    provided that one can perform $N$\! Bell measurements on $\rho\otimes\rho$, with
    \begin{equation}
        N = \BigO\mathclose{}\Big(\frac{\QubitsN^3 2^{4\chi} 2^{2M_0}}{\Delta^4} \log\frac{1}{\delta}\Big).
    \end{equation}

    Producing a single sample takes time $\BigO(N \QubitsN^2)$ time in classical
    postprocessing of the obtained Bell measurement data.
\end{theorem}

\section{Approximate Sample and Query}

We define the \textit{approximate sample and query} model as a modification of the \textit{sample and query} model (sec.~\ref{sec:sq}). The main observations motivating these modifications are
that, even given arbitrarily many copies of some general state, one cannot cheaply
compute a given amplitude to exponential precision; that the required precision will
dictate the cost of the operation; and that quantum operations may fail depending on the
operation, with one-sided or two-sided error.

\begin{definition}{Approximate query access} \label{def:aq}
    We say that we have Approximate Query (AQ) access to a vector $x \in \Complex^d$ if,
    for any choice of $\epsilon \in (0, 1]$, and for any $i \in [d]$, one may obtain an
    estimate of $x(i)$ to absolute error $\epsilon$, with two-sided error probability of
    at most $1/3$, and runtime $\Query(\epsilon)$.
\end{definition}

\begin{definition}{Approximate sample access} \label{def:as}
    We say that we have Approximate Sample (AS) access to a vector $x \in \Complex^d$ if
    we may sample from a probability distribution $\Distribution_{x}$, defined by its
    probability mass function as
    \begin{equation}
        \Distribution_x(i) = \frac{\abs{x(i)}^2}{\norm{x}^2} \mathrlap{\,,\qquad i=1,2,\ldots,d}
    \end{equation}
    with one-sided error probability at most $1/3$ and runtime $\Sample$, and
    $\norm{x}^2$ may be obtained to additive error $\epsilon$, with at most two-sided
    error probability $1/3$, and runtime $\QNorm(\epsilon)$. (By one-sided error probability
    over the sampling, we mean that either a sample is successfully produced, or the
    algorithm emits a flag indicating failure.)
\end{definition}

The definition of Approximate \textit{sample and query} then combines these two forms of access:
\begin{definition}{Approximate sample and query} \label{def:asq}
    Say that we have \textit{approximate sample and query} (ASQ) access to a vector $x \in
        \Complex^d$ if we have AS and AQ access to it.
\end{definition}

The corresponding definition for oversampling follows naturally:
\begin{definition}{\textit{Approximate $\phi$-oversample and query} accces} \label{def:aosq}
    For $\phi \geq 1$, say that we have Approximate $\phi$-Oversample and Query
    ($\ASQ[\phi]$) access to a vector $x \in \Complex^d$ if we have AQ access to $x$
    (Definition~\ref{def:aq}) and we have AS access to $\tilde{x} \in \Complex^d$
    (Definition~\ref{def:as}) satisfying that, for all $i \in [d]$, $\abs{\tilde{x}(i)}
        \geq \abs{x(i)}$ and $\norm{\tilde{x}}^2 \leq \phi \norm{x}^2$. [Recall that the time
            to query an entry of $x$ to error $\epsilon$ is $\Query(\epsilon)$, the time to query
            $\norm{\tilde{x}}^2$ to error $\epsilon$ is $\QNorm(\epsilon)$, and that the time
            to sample from $\Distribution_{\tilde{x}}$ is $\Sample$.]
\end{definition}
Compare this definition with that of Def.~\ref{def:osq}: we are careful to require only
\emph{sample} access to $\tilde{x}$, whereas \textit{oversample and query} asks for sample
\emph{and} query access to $\tilde{x}$.

\subsection{Satisfying approximate sample and query access}
\label{sec:satsf-asq}

The results of van~Apeldoorn \et al.~\cite{vanApeldoorn2023}, discussed in section~\ref{sec:single-copy-tom},
together with the standard properties of quantum states and measurements suffice to see
that access to single-copy preparations of a quantum state suffices to satisfy ASQ access
to the underlying quantum state vector (when represented in the computational basis).
Because we only require one amplitude at a time, we may reduce the necessary quantum time
after the state preparation:

\begin{lemma} \label{lemma:estm-phase}
    Let $\ket{\psi}$ be an $\QubitsN$ qubits state, prepared in time $T$ in black-box
    fashion. Then, $\BigO[T \epsilon^{-2} \QubitsN]$ time suffices to estimate
    $\braket{i}{\psi}$ for any choice of $i \in [2^\QubitsN]$ (up to a global phase, but
    with phases consistent across estimates), to error $\epsilon$ and with two-sided error
    probability of at most $1/3$. Furthermore, at most $\BigO(T + \QubitsN)$ quantum time
    [and a $\BigO(\log\QubitsN)$ depth overhead in the quantum circuits] is needed between
    measurements.
\end{lemma}
\begin{proof}
    The proof is a modification of that of Theorem~\ref{thm:vad-tomography}. Employing
    Theorem~\ref{thm:smpl-estim}, retrieve an estimate of $\abs{\braket{i}{\psi}}$ for all
    $i \in [2^\QubitsN]$ up to $\infty$-norm error $\epsilon/16$ and with error
    probability $1/9$. This takes time $\BigO[T\epsilon^{-2}\QubitsN]$, and the resulting
    vector is $k$-sparse, where $k = \BigO(\epsilon^{-2}\QubitsN)$. Denote this vector by
    $r$. In time $\BigO(k)$, round any entry of this vector below $\epsilon/2$ to zero.
    Then, in time $\BigO(k)$, determine $m \in [2^\QubitsN]$ such that $r(m)$ is the
    greatest entry of $r$. Now, for given $j \in [2^\QubitsN]$, if $r(j) = 0$, estimate
    $\braket{j}{\psi}$ as $0$. If $j = m$, estimate $\braket{j}{\psi}$ as $r(m)$.
    Otherwise, write unitaries $V$ and $V'$ that map $\ket{m} \mapsto \ket{0}$ and
    $\ket{i} \mapsto \ket{1}$ and that then apply a Hadamard gate on the least significant
    bit, in the case of $V$, or then apply a phase gate followed by a Hadamard gate on the
    least significant bit, in the case of $V'$. This is done in time $\BigO(\QubitsN)$ and
    produces a circuit with $\BigO(\QubitsN)$ gates. Measuring $V\ket{\psi}$ and
    $V'\ket{\psi}$ $\BigO(\epsilon^{-2})$ times (for a total $\BigO[(T + \QubitsN)
    \epsilon^{-2}]$ runtime) suffices to estimate both $\abs{\braket{m}{\psi} -
    \braket{j}{\psi}}$ and $\abs{\braket{m}{\psi} - \I\braket{j}{\psi}}$ to error
    $\epsilon/16$ with joint success probability at least $2/9$. Let these estimates be
    denoted by $s$ and $t$ respectively. Then, $[r(m)^2 + r(j)^2 - s^2]/[2r(j)] + \I
    [r(m)^2 + r(j)^2 - t^2]/[2r(j)]$ is an $\epsilon$-error estimate of $\braket{j}{\psi}$
    with two-sided error probability at most $1/3$. The depth overhead is given by the
    depth of the circuits implementing $V$ and $V'$, which are the same up to a difference
    of $1$. The depth of $V$ is the same as that of a circuit mapping
    $\ket{0}\mapsto\ket{0}, \ket{m\oplus i}\mapsto\ket{1}$, up to a difference of $1$ (and
    where $\oplus$ denotes bitwise \textsc{xor}). Since $\BigO(1)$ layers of controlled
    \textsc{not} gates can map $\ket{u}$ to $\ket{v}$ such that $\abs{v} = \abs{u}/2$
    (while mapping $\ket{0}$ to $\ket{0}$), and since $\abs{m\oplus i}\leq n$, the depth
    of $\BigO(\log n)$ follows.
\end{proof}

Then, the following remark follows:
\begin{remark} \label{remark:asq-prepmeasure}
    Given the ability to prepare a quantum state $\ket{\psi}$ of $\QubitsN$ qubits in time
    $T$ (in black-box fashion), and the ability to run quantum circuits of up to time
    $\BigO(T + \QubitsN)$ [or the ability to run quantum circuits $\BigO(\log\QubitsN)$
    deeper], one has ASQ access (Definition~\ref{def:asq}) to $x \in
    \Complex^{2^\QubitsN}$, where $x(i) = \braket{i}{\psi}$, up to a global phase on
    $\ket{\psi}$. The runtime costs are $\QNorm(\epsilon) = 0$, $\Sample = T$,
    $\Query(\epsilon) = \BigO(T \epsilon^{-2} \QubitsN)$. More generally, given the
    ability to prepare an $\QubitsM$-block encoding of a (subnormalized) state
    $\ket{\phi}$ of $\QubitsN$ qubits in time $T$ (in black-box fashion), and the ability
    to run circuits of up to time $\BigO(T + \QubitsN)$, one has ASQ access to $x \in
    \Complex^{2^\QubitsN}$, where $x(i) = \braket{i}{\phi}$, up to a global phase on
    $\ket{\phi}$. In this case, the runtime costs are $\QNorm = \BigO(\epsilon^{-2} T)$,
    $\Sample = \BigO(T\norm{x}^{-2})$, $\Query(\epsilon) = \BigO(T\epsilon^{-2}\QubitsN)$.
\end{remark}

Note that the procedures outlined in Theorem~\ref{thm:smpl-estim} and
Lemma~\ref{lemma:estm-phase} apply unchanged to the block encoding setting, from whence
the above remark specifies $\Query(\epsilon)$ for block encodings. On the other hand, the
sample time $\Sample$ for block encodings corresponds to a single successful
post-selection. Because successful post-selection happens with probability $\norm{x}^2$
(from the standard properties of quantum measurement), and from a Chebyshev inequality,
the stated sample time $\Sample$ follows. (Further, these do not require that $\norm{x}$
is known.) From this follows also the norm query time $\QNorm$.

The remark above does not consider $(\QubitsM, \epsilon)$-block encodings: clearly,
access to a $(\QubitsM, \epsilon)$-block encoding of a vector $x$ gives \textit{approximate
sample and query} access to $x'$ satisfying $\norm{x' - x} \leq \epsilon$. While the
consequences of this for query access are obvious --- one has query access to $x$ up to
error $\epsilon$ --- the effects on \textit{approximate sample access} are made more meaningful by
the following lemma (also observed in \cite[Lemma 4.1]{Tang2019} for the particular case
of the reals, and for which the proof proceeds identically):
\begin{lemma} \label{lemma:approx-tvd}
    Let $x,y \in \Complex^d$ be two vectors such that $\norm{x - y} \leq \epsilon$ for
    some $\epsilon > 0$.
    Let $\Distribution_x$ and $\Distribution_y$ be the distributions with probability mass
    function $\Distribution_x(i) = {\abs{x(i)}}/{\norm{x}^2}$, and $\Distribution_y(i) =
        {\abs{y(i)}^2}/{\norm{y}^2}$. Then, the distance between these two distributions is
    bounded as
    \begin{equation}
        \sum_{i \in [d]} \abs{\Distribution_x(i) - \Distribution_y(i)} \leq \frac{4\epsilon}{\norm{x}}.
    \end{equation}
\end{lemma}

Thus, we have shown that \textit{sample and query} access captures at least some of the power of
the ability to prepare and measure (singly coherent) copies of a block encoding of a
state.

However, the definition allows us to satisfy ASQ access in other contexts.
For example, leveraging Theorem~\ref{thm:t-gate-sim} (cf.~section~\ref{sec:low-t}), we
note that, given a low $T$-gate count circuit describing a quantum state $x$, we obtain a
situation quite similar to that of $(\QubitsM, \epsilon)$-block encodings:
\begin{remark} \label{remark:asq-simt}
    Given a classical description of a quantum state $\ket{\psi}$ in the terms of
    Theorem~\ref{thm:t-gate-sim}, thus by a circuit with at most $t$ $T$-gates, and
    denoting $x$ as the representation of $\ket{\psi}$ in the computational basis, one has
    approximate sample access to $x'$ satisfying $\norm{x - x'} \leq \Delta/4$ at runtime
    cost $\QNorm(\epsilon) = 0$ and $\Sample = \tilde\BigO[\Poly(n)(1 + 2^t
            \Delta^{-4})]$. Furthermore, for choice of $i \in [2^\QubitsN]$, the circuit
    performing a swap test between $U\ket{0}^{\otimes\QubitsN}$ and $\ket{i}$ has
    $\BigO(\QubitsN)$ qubits and $\BigO(t)$ $T$-gates; thus also one has approximate query
    access to $x$ at runtime cost $\BigO[\Poly(n)(1 + 2^t \epsilon^{-2})]$. This does not
    require quantum computation.
\end{remark}

Finally, the \textit{approximate sample and query} access model is sufficiently general that it
also captures some of the power of measuring a quantum state in a completely different
basis; namely, the Pauli string basis:\footnote{We note here the parallel to
Ref.~\cite{vdNest2011}'s ``rotated bases computationally tractable states''.}

\begin{remark} \label{remark:asq-pauli}
    Given the ability to prepare copies of a quantum state $\ket{\psi}$ of dimension
    $2^\QubitsN \equiv d$ in time $T$, and if performing Pauli sampling on $\ket{\psi}$
    takes time $T'$, then one has ASQ access to $\pi_{\dyad{\psi}}$
    [cf.~eq.~\eqref{eq:pauli-vec}], at runtime costs $\QNorm(\epsilon) = 0$, $\Sample =
        T'$, and $\Query(\epsilon) = \BigO(Td^{-1}\epsilon^{-2})$.
\end{remark}

Note that the query runtime cost follows from the fact that, having chosen a given Pauli
string $P_i$, rotated single-qubit measurements suffice to estimate
$\ExpVal{\psi}{P_i\mkern1mu}$. From Ref.~\cite{Okamoto1959}, we may attain error
$\epsilon$ with $\BigO(\epsilon^{-2})$ samples; and, from eq.~\eqref{eq:pauli-vec}, error
$\epsilon/\sqrt{d}$ suffices. This procedure requires at most $\BigO(T + \QubitsN)$
quantum time between measurements.

\subsection{Composing approximate oversample and query access}

One significant property of the \textit{sample and query} model is that it can be composed,
meaning that \textit{sample and query} access to separate objects (say, two vectors $x$ and $y$)
can be leveraged to obtain \textit{(over-)sample and query} access to arithmetic over those
objects (say, $x+y$). This is the property enabling partial dequantization of linear
algebra algorithms in Ref.~\cite{Chia2020}. While
the proofs in Ref.~\cite{Chia2020} for composition do not translate to the \textit{approximate
oversample and query} model (as they rely on the ability to query to arbitrary
precision), we now illustrate that \textit{approximate oversample and query} access retains
composability. We do not show a result as strong as those stated for \textit{sample and
query}, but take a first step in this direction. Since the ability to perform linear
combinations of quantum states given the ability to prepare each state individually is
already a significant operation in quantum computation \cite{Childs2012}, we show that
$\ASQ[\phi]$ composes over linear combinations.

\begin{lemma} \TheoremName{Linear combinations from $\ASQ[\phi]$}
    \label{lemma:lin-comb}
    Let $\phi \geq 1$. Given $\ASQ[\phi]$ to $\tau$ vectors $x_1, x_2, \ldots, x_\tau \in
    \Complex^d$, and $\tau$ nonzero complex coefficients $\lambda_1, \ldots, \lambda_\tau
    \in \Complex$, one may construct $\ASQ[\phi']$ access to the linear combination $u
    \Defined \sum_{j \in [\tau]} \lambda_j x_j$, where $\phi' = \phi \tau^2 \kappa^2$,
    and $\kappa$ is the spectral condition number of the matrix whose columns are the
    vectors $x_1, \ldots, x_\tau$. The costs of access are:
    \begin{gather}
        \Query(\epsilon) = \BigO[\tau \log \tau \Query_x(\epsilon/\norm{\lambda})] \\
        \QNorm(\epsilon) = \BigO\{ \tau \log \tau \QNorm_x[\epsilon/(\tau^2 \norm{\lambda}^2)] \} \\
        \Sample = \BigO(\tau + \Sample_x),
    \end{gather}
    where $\lambda = (\lambda_1, \ldots, \lambda_\tau)$, and $\Query_x$, $\QNorm_x$, and
    $\Sample_x$ are upper bounds to the costs defined in Definition~\ref{def:aosq} for
    all $x_1, \ldots, x_\tau$.
\end{lemma}
\begin{proof}
    Query access to $u$ is attained simply by querying each entry to sufficient precision;
    namely, for $j\in[\tau]$, $i\in[n]$, if $x_j(i)$ is known to precision
    $\epsilon/\abs{\lambda_j}$, $u(i)$ is known to precision $\epsilon$. With the median
    of $18\log(6\tau)$ copies of each estimate (considering the real and imaginary part of
    each estimate separately), one also ensures that all estimates are correct with
    overall probability $2/3$. [Note that this requires each $x_j(i)$ to be estimated to
            precision $\epsilon/(\sqrt{2}\abs{\lambda_j})$, but from this results only a constant
            factor.]

    To show that oversampling in the stated conditions is also possible, we
    propose an explicit procedure and then argue that the oversampling factor $\phi'$ is as
    defined. Let $\tilde{u}$ be the ``oversampling vector'' for $u$, such that we should
    satisfy sample access to $\tilde{u}$ (cf.\ also Definition~\ref{def:aosq}). Then
    sample from $\tilde{u}$ by first choosing $j$ with probability
    $\abs{\lambda_j}^2/\norm{\lambda}^2$ and outputting $i$ as sampled from approximate
    $\phi$-oversampling of $x_j$. This corresponds to sampling from
    $\Distribution_{\tilde{u}}$, where
    \begin{equation}
        \tilde{u}(i) = \sqrt{\tau (\sum_{j\in[\tau]}\norm{\tilde{x}_j}^2) \sum_{k\in[\tau]} \abs{\lambda_k}^2 \frac{\abs{\tilde{x}_k(i)}^2}{\norm{\tilde{x}_k}^2}}.
    \end{equation}
    For all $i\in[d]$, $\abs{\tilde{u}(i)} \geq \abs{u(i)}$, by a triangle inequality and
    a geometric mean--harmonic mean inequality, and because $\tau(\sum_{j\in[\tau]}
    \norm{\tilde{x}_j}^2) \geq \tau \norm{\tilde{x}_k}^2$ for all $k\in[\tau]$. Thus this
    procedure samples from $\Distribution_{\tilde{u}}$, which is oversampling
    $\Distribution_u$. The oversampling factor $\phi'$ is bounded, by definition, by an
    upper bound to $\norm{\tilde{u}}^2/\norm{u}^2$. Write $X$ as the ($d\times\tau$)
    matrix whose columns are $x_1, \ldots, x_\tau$, and note that $\sum_{j\in[\tau]}
    \norm{x_j}^2 \equiv \norm{X}_F^2 \equiv \Tr(X^\dag X)$. Then, we have that
    \begin{equation}
        \phi' \leq \frac{\phi \tau \norm{X}_F^2 \norm{\lambda}^2}{\norm{X \lambda}^2} \leq \frac{\phi\tau}{\sigma_{\text{min}}(\frac{X}{\norm{X}_F})^2} \leq \phi\tau^2\kappa^2
    \end{equation}
    where $\sigma_{\text{min}}(A)$ denotes the smallest singular value of $A$, and
    $\kappa$ denotes the ratio of the largest to the smallest singular value of $X$. The
    first inequality follows from $\norm{\tilde{x}_j}^2 \leq \phi \norm{x_j}^2$, the
    second inequality by a ``min-max'' formulation of the singular values, and the third
    inequality follows from from the identity $\Tr(X^\dag X) = \sum_{i \in [\Rank X]}
    \sigma_i^2$, if $\sigma_i$ is the $i$th singular value of $X$. Since the outlined
    procedure requires a single sample, it correctly fails with one-sided $1/3$
    probability error.

    Finally, estimating the norm $\norm{\tilde{u}}^2$ requires only that every norm
    $\norm{\tilde{x}_j}^2$ is known to error $\epsilon/(\tau^2\norm{\lambda}^2)$. With
    the median of $18\log(3\tau)$ copies of each estimate for $\norm{\tilde{x}_j}$ (as
    enabled by the $\ASQ[\phi]$ access to each $\tilde{x}_j$), one can ensure that every
    estimate is correct with two-sided $1/3$ probability error.
\end{proof}

\bigglue

The dependence on the condition number $\kappa$ in the lemma above is not a problem if
the various $x_1, \ldots, x_\tau$ are approximately orthogonal and not too different in
magnitude. But this is a rather restrictive condition. Indeed, the conditions stated by
the ASQ model regarding the error ``sidedness'' for each operation prevent an approach
exactly like that of Ref.~\cite[Lemma 3.9]{Chia2020}.\footnote{ One finds that many
proofs regarding the SQ model rely on exactly querying a given entry or the norm of the vector to cancel out undesirable terms. The introduction of error \emph{and}
randomness impedes the same approach.} If we are willing to fail the construction with
some probability, then an analogous result is recovered, and the dependence on such a
$\kappa$ vanishes.

We first require a lemma on obtaining relative error estimates from absolute error:

\begin{lemma} \label{lemma:rel-err}
    For choice of $\rho\in(0,1]$, and given the ability to estimate a quantity $x \in
        \Complex$ to absolute error $\epsilon$ in time $\Query(\epsilon)$ with probability
    $2/3$, algorithm~\ref{alg:rel-err-estim} computes an estimate of $x$ to relative error
    $\rho$ with probability $1-\delta$. The algorithm halts with overwhelming probability
    in time $\BigO[\Query(\rho \abs{x}) (1 + \log\frac{1}{\abs{x}}) \log
            \frac{1}{\delta}]$ [under the assumption that $\Query\{\Omega(\epsilon)\} =
                \BigO\{\Query(\epsilon)\}$].
\end{lemma}
\begin{proof}
    Algorithm~\ref{alg:rel-err-estim} computes increasingly accurate estimates of $x$,
    looking to halt when the estimate's absolute value (say, $\abs{\hat{x}}$) is
    twice as large as the estimation error (say $\epsilon$). This ensures that $\epsilon
        \leq \abs{x}/3$, and that $\frac{2}{3}\abs{x} \leq \abs{\hat{x}} \leq 2\abs{x}$,
    i.e., that $\abs{\hat{x}}$ is of the order of $\abs{x}$. Once this bound for
    $\abs{x}$ is obtained (say, $\bot$), it suffices to output an estimate of error
    $\rho\bot$.

    To deal with the randomization in the estimations of $\abs{\hat{x}}$, we use the fact
    that, for $w \in \Real$ and $\hat{w}$ such that $\abs{w - \hat{w}} \leq \epsilon$
    with probability $2/3$, the median of $18\log(1/\delta)$ independent copies of
    $\hat{w}$ is within $\epsilon$ error of $w$ with probability $1-\delta$ (from
    Hoeffding's inequality). This allows the algorithm to guarantee that: false positives
    do not occur within the loop with more than (overall) probability $\delta/2$; and that
    the loop is not kept running by bogus estimates for too long. Let $k_{\textsc{max}}$
    be the maximum number of iterations of the loop needed to achieve a sufficiently small
    value of $\epsilon$: from the considerations above, $k_{\textsc{max}} = \lceil
    \log_2(3/\abs{x}) \rceil$. After $k$ exceeds this value, the probability of repeatedly
    looping from bad estimates vanishes overwhelmingly: the probability of reaching the
    $(k_{\textsc{max}} + m)$th iteration of the loop is $(\delta/2^{k_{\textsc{max}} + 1})
    (\delta /2^{k_{\textsc{max}} + 2})\cdots(\delta/2^{k_{\textsc{max}}+m})$, from which
    it is easy to see that with probability at least $1-\omega$, the loop halts at the
    $(k_{\textsc{max}} + m)$th iteration, where $m = \log_2(1/\omega)/[1 +
    k_{\textsc{max}} + \log_2(1/\delta)]$. Setting $\delta = \delta'\omega$, this ensures
    that the loop halts in $k_{\textsc{max}}$ iterations with probability $1-\omega$. (If
    the loop does not halt in $k_{\textsc{max}}$, it halts soon thereafter with
    overwhelming probability, from the same considerations.) From this observation, the
    explicit form of the algorithm, and the value of $k_{\textsc{max}}$ follows the stated
    runtime.
\end{proof}

\begin{figure*}
    \begin{algorithm}[H]
        \caption{
            Calculating a relative error estimate via online access to randomized queries
            to absolute error estimates.}
        \label{alg:rel-err-estim}

        \def\Loop{\Statex \!\texttt{loop:}}
        \def\Final{\Statex \!\texttt{final:}}
        \algnewcommand{\GoTo}[1]{\State \textbf{goto} \texttt{#1}}
        \algblockdefx{DoTimes}{EndDoTimes}[1]{\textbf{do} #1 \textbf{times:}}{\textbf{end}}

        \begin{algorithmic}[1]
            \Require relative precision $\rho$
            \Require failure probability $\delta$
            \Require online access to $\hat{x}(\epsilon)$ such that $\abs{\hat{x}(\epsilon) - x} \leq \epsilon$ with probability $2/3$
            \algstore{alg:rel-err-estim-a}
        \end{algorithmic}
        \begin{varwidth}[t]{.5\linewidth}
            \begin{algorithmic}[1]
                \algrestore{alg:rel-err-estim-a}
                \Statex\vskip-2pt
                \State $k \gets 1$
                \Loop
                \State $\epsilon \gets 2^{-k}/\sqrt{2}$
                \State $l \gets [\textsc{empty list}]$
                \DoTimes{$18\log[10^4 \, 2^{k+1}/\delta]$}
                \State Sorted insert $\abs{\hat{x}(\epsilon)}$ into $l$
                \EndDoTimes
                \State $\mu \gets \textsc{median}(l)$
                \If{$\epsilon \leq \frac{1}{2}\mu$}
                \GoTo{final}
                \EndIf
                \algstore{alg:rel-err-estim-b}
            \end{algorithmic}
        \end{varwidth}
        \hfil
        \begin{varwidth}[t]{.5\linewidth}
            \begin{algorithmic}[1]
                \algrestore{alg:rel-err-estim-b}
                \Statex
                \State $k \gets k+1$
                \GoTo{loop}
                \Final
                \State $\bot \gets \mu - \epsilon$
                \State $r \gets [\textsc{empty list}]$
                \State $c \gets [\textsc{empty list}]$
                \DoTimes{$18 \log[8/\delta]$}
                \State Sorted insert $\Re[\hat{x}(\rho\bot)]$ into $r$
                \State Sorted insert $\Im[\hat{x}(\rho\bot)]$ into $c$
                \EndDoTimes
                \State\Return $\textsc{median}(r) + \I \, \textsc{median}(c)$
            \end{algorithmic}
        \end{varwidth}
        \vskip 2pt
    \end{algorithm}
\end{figure*}

With this lemma, we may conclude an alternative linear combination composition lemma that
does not have a condition number dependence and is very similar to \cite[Lemma
3.9]{Chia2020}, at the cost of succeeding only probabilistically:

\begin{lemma} \TheoremName{Probabilistic linear combinations from $\ASQ[\phi]$}
    \label{lemma:lin-comb-rand}
    Let $\phi \geq 1$ and $\delta \in (0, 1)$. Given $\ASQ[\phi]$ to $\tau$ vectors $x_1,
    \ldots, x_\tau \in \Complex^d$, and $\tau$ nonzero complex coefficients $\lambda_1,
    \ldots, \lambda_\tau$, one may construct $\ASQ[\phi']$ access to the linear
    combination $u \Defined \sum_{j\in[\tau]} \lambda_j x_j$ with probability at least
    $1-\delta$, where $\phi' = 4\tau\phi (\sum_{k\in[\tau]} \norm{\lambda_k
    x_k}^2)/\norm{u}^2$. The costs of access are:
    \begin{samepage}
        \begin{gather}
            \Query(\epsilon) = \BigO[\tau \log \tau \Query_x(\epsilon/\norm{\lambda})] \\
            \QNorm(\epsilon) = \BigO(1) \\
            \Sample = \BigO(\tau + \Sample_x),
        \end{gather}
    \end{samepage}
    where $\lambda = (\lambda_1, \ldots, \lambda_\tau)$, and $\Query_x$, $\QNorm_x$, and
    $\Sample_x$ are upper bounds to the costs defined in Definition~\ref{def:aosq} for
    all $x_1, \ldots, x_\tau$. For $\norm{x}^2$ the smallest $\norm{x_j}^2$, $j=1,
    \ldots, \tau$, this construction requires a one-time preliminary procedure that runs
    in time $\BigO[\tau \QNorm(\norm{x}^2) (1 + \log \frac{1}{\norm{x}})
    \log\frac{\tau}{\delta}]$ [with overwhelming probability and under the assumption
    that $\QNorm\{\Omega(\epsilon)\} = \BigO\{\QNorm(\epsilon)\}$].
\end{lemma}
\begin{proof}
    Using Lemma~\ref{lemma:rel-err}, start by computing a $1/2$-relative error estimate of
    each norm $\norm{\tilde{x}_j}^2$, $j=1,\ldots,\tau$, with overall error probability
    $\delta$, and store these estimates as $\hat{n}_1^2, \ldots, \hat{n}_\tau^2$. The
    runtime for this follows directly from Lemma~\ref{lemma:rel-err} and is as
    stated therein. The construction fails only if any of these estimates are
    wrong; henceforth we assume these are correct.

    We otherwise proceed similarly to \ref{lemma:lin-comb}: estimating each $x_j(i)$ to
    error $\epsilon/\abs{\lambda_j}$ suffices, with the median of means being used to ensure
    every estimate is correct with probability $2/3$.

    Sampling is done by first choosing $j$ with probability $\abs{\lambda_j}^2
    \hat{n}_j^2 / (\sum_{k\in[\tau]} \abs{\lambda_k}^2 \hat{n}_k^2)$ and outputting $i$
    as sampled from approximate $\phi$-oversampling of $x_j$. This corresponds to
    sampling from $\Distribution_{\tilde{u}}$, where $\tilde{u}$ is defined by its
    entries as
    \begin{equation}
        \tilde{u}(i) = \sqrt{2\tau \sum_{k\in[\tau]} \abs{\lambda_k}^2 \hat{n}_k^2 \frac{\abs{\tilde{x}_k(i)}^2}{\norm{\tilde{x}_k}^2}}.
    \end{equation}
    This is for all $i \in [d]$, $\abs{\tilde{u}(i)}^2 \geq \abs{u(i)}^2$, by a triangle
    and a geometric mean--harmonic mean inequality, and making use of the fact that, by
    hypothesis and for all $j \in [\tau]$, $\frac{1}{2}\norm{\tilde{x}}^2 \leq \hat{n}_j
    \leq \frac{3}{2}\norm{\tilde{x}}^2$. To bound $\phi'$, note that $\norm{\tilde{u}}^2
    = 2\tau \sum_{k\in[\tau]} \abs{\lambda_k}^2 \hat{n}_k^2$, and again make use of the
    bounds of $\hat{n}_j$, and the fact that $\norm{\tilde{x}_j}^2 \leq
    \phi\norm{x_j}^2$. Furthermore, this explicit norm can be calculated once as a
    preliminary step, after which computing it corresponds to a lookup.
\end{proof}

\bigglue

\subsection{Computation with approximate oversample and query access}
\label{sec:computation}

An important characterization of the ASQ model is its computational power: it is
only meaningful to say that one can satisfy the ``promises'' of ASQ if one knows how to
perform useful computation by leveraging those promises. Here, we focus on the task of
performing an inner product estimation between two vectors, given ASQ access to either
one or both of the vectors. The computational cost of computing the inner product between
two states, given the ability to prepare singly coherent copies of each, is recognized as
significant (for example, for fidelity estimation) and is under ongoing study
\cite{Anshu2022, Hinsche2024, Arunachalam2024, Gong2024}. On the other hand, the inner
product can be regarded as the simplest nontrivial ``reduction'' of two vectors into a
scalar to estimate and is also considered in Refs.~\cite{Tang2019,Chia2020,LeGall2023}.
This task will find immediate application in the application of \ASQ\space to distributed
inner product estimation (Sec.~\ref{sec:dist-inprod})

We begin with the asymmetric case, i.e., \textit{approximate sample and query} access to one
vector, and classical oracle access to the entries of another:

\begin{lemma} \TheoremName{Asymmetric inner product estimation}
    \label{lemma:inprod-asym}
    Let $x,y \in \Complex^d$, and $\epsilon \in (0, 1]$. Suppose $\ASQ[\phi]$ access to
    $x$ and query access to $y$ (in the language of Ref.~\cite{Chia2020};
    cf.~Def.~\ref{def:sq}). Then, one may estimate $x^\dagger y$ to additive error
    $\epsilon\norm{y}_1$ and success probability $2/3$ by
    algorithm~\ref{alg:inprod-asym}. The runtime of the algorithm is
    $$\tilde\BigO[\phi\epsilon^{-2}\norm{x}^2\log d + \phi\epsilon^{-2}\norm{x}^2\Query(\epsilon) + \phi^2\epsilon^{-4}\norm{x}^4 \log d \,\Sample + \QNorm(\norm{x}^2)], $$
    with very high probability, and under the assumption that $\Query[\Omega(\epsilon)] =
    \BigO[\Query(\epsilon)]$ and likewise for $\QNorm$.
\end{lemma}
\begin{figure*}
    \begin{algorithm}[H]
        \caption{
            Calculating an inner product estimate via $\ASQ[\phi]$ access to one of the
            vectors ($x$), and classical RAM access to the other ($y$), with success
            probability at least $2/3$, and absolute error at most $\epsilon\norm{y}_1$.}
        \label{alg:inprod-asym}

        \algblockdefx{DoTimes}{EndDoTimes}[1]{\textbf{do} #1 \textbf{times:}}{\textbf{end}}

        \begin{varwidth}[t]{.5\textwidth}
            \begin{algorithmic}[1]
                \Require $\ASQ[\phi]$ access to $x \in \Complex^d$
                \Require Query access to $y \in \Complex^d$
                \Require Precision $\epsilon \in (0, 1]$
                \Statex
                \Function{Improve}{$x$, $\epsilon$, $\delta$}
                \State $r \gets []$
                \State $c \gets []$
                \State $\epsilon \gets \epsilon/\sqrt{2}$
                \DoTimes{$\lceil 18\log(1/\delta)\rceil$}
                \State Sorted insert $\Re[\hat{x}(\epsilon)]$ into $r$
                \State Sorted insert $\Im[\hat{x}(\epsilon)]$ into $c$
                \EndDoTimes
                \State\Return $\textsc{median}(r) + \I\, \textsc{median}(c)$
                \EndFunction
                \Statex
                \State $\hat{n}^2 \gets \Call{[algorithm~\ref{alg:rel-err-estim}]}{\norm{x}^2, \rho=\frac{1}{4}, \delta=\frac{1}{9}}$
                \State $\gamma \gets \epsilon^2/(135\,{\hat{n}^2})$
                \algstore{alg:inprod-asym}
            \end{algorithmic}
        \end{varwidth}
        \hfil
        \begin{varwidth}[t]{.5\textwidth}
            \begin{algorithmic}[1]
                \algrestore{alg:inprod-asym}
                \State $h \gets [\textsc{empty histogram}]$
                \DoTimes{$\lceil 512 \, \gamma^{-2} \log(18d) \rceil$}
                \Repeat
                \State $i \gets \Call{oversample}{x}$
                \Until{sample is valid.}
                \State Insert $i$ into $h$.
                \EndDoTimes
                \Statex
                \State $m \gets []$
                \DoTimes{$\lceil 7\hat{n}^2 \epsilon^{-2} \rceil$}
                \Repeat
                \State $i \gets \Call{oversample}{x}$
                \Until{sample is valid.}
                \State $\hat{p} \gets \Call{frequency}{h; i}$
                \If{$\hat{p} \geq 3\gamma/2$}
                \State $\hat{x} \gets \textsc{improve}[x(i), \epsilon/4, \epsilon^2/(127\,\hat{n}^2)]$
                \State Insert $(\hat{p} + \gamma/2)^{-1} \hat{x}^* y(i)$ into $m$.
                \EndIf
                \EndDoTimes
                \State \Return $\Call{average}{m}$
            \end{algorithmic}
        \end{varwidth}
        \vskip 2pt
    \end{algorithm}
\end{figure*}
\begin{proof}
    Broadly speaking, Algorithm~\ref{alg:inprod-asym} computes an approximation to every
    probability using Theorem~\ref{thm:smpl-estim}. Then, it filters the values of
    $i$, as sampled from the ASQ access to the input vector $x$, excluding entries for
    which $x(i)$ is small. By using a Chebyshev inequality, and showing that the
    bias and variance of the considered point estimation are sufficiently bounded, one can
    ensure that the average of a sufficient number of samples is $\epsilon\norm{y}_1$
    close to the true inner product. This is in line with Refs.~\cite{Tang2019, Chia2020,
        LeGall2023}; but one now needs to adapt the filtering technique and the bias and
    variance bounding to the presence of error in the knowledge of the entries of $x$.

    We now discuss the proof in detail.

    Denote the ``oversampling vector'' of $x$ as $\tilde{x}$ (cf.~Def.~\ref{def:aosq}).
    Let $\hat{n}^2$ be an estimate of $\norm{\tilde{x}}^2$ to relative error $1/4$ with
    probability $8/9$, as made possible by $\ASQ[\phi]$ to $x$ and
    Lemma~\ref{lemma:rel-err}. This is attained, accordingly, in time
    $\BigO\{\QNorm(\norm{\tilde{x}}^2) [1 + \log(1/\norm{\tilde{x}}^2)]\}$, and we
    henceforth assume the estimate is correct. Define $\gamma\Defined\epsilon^2/(135 \,
    \hat{n}^2)$. Taking also $\BigO(\gamma^{-2}\log d)$ samples from
    $\Distribution_{\tilde{x}}$, and from Theorem~\ref{thm:smpl-estim}, one knows also
    every $\Distribution_{\tilde{x}}(i)$ to error $\gamma/8$, with overall probability
    $8/9$. [Looking up each estimate of $\Distribution_{\tilde{x}}(i)$, say, $\hat{p}_i$,
    takes time $\BigO(\log d)$, per the Theorem.]

    We also note the following preliminary observations:
    \begin{enumerate}[(a),wide,nosep] \DontBreakHereButBreakBetweenItems
        \item $\norm{\tilde{x}}/\hat{n} \leq 2$;
        \item $\abs{x(i)} \leq \abs{\tilde{x}(i)} = \norm{\tilde{x}} \sqrt{\TXDist(i)}$, for all $i\in[d]$; and
        \item $\abs{a/(\hat{a} + \epsilon) - 1} \leq 2\epsilon/a$, for all positive $a$
              and $\hat{a}$ satisfying $\abs{a - \hat{a}} \leq \epsilon$, as a direct
              consequence of this fact.
    \end{enumerate}

    Now, consider the following point estimator, if $i \sim \Distribution_{\tilde{x}}$,
    and $\hat{x}_i$ is an $\epsilon/4$ estimate of $x(i)$:
    \begin{equation}
        \widehat{x^\dagger y}(i) = \begin{cases}
            0                                               & \text{if } \hat{p}_i \leq \frac{3}{2} \gamma \\
            \frac{1}{\hat{p}_i + \gamma/8} \hat{x}_i^* y(i) & \text{otherwise.}
        \end{cases}
    \end{equation}
    Let us say that the first case corresponds to the ``rejection'' of $i$, and the second
    corresponds to the ``acceptance'' of $i$. From the error bounds on $\hat{p}_i$, for
    all rejected $i$, one has that $\Distribution_{\tilde{x}}(i) \leq 2\gamma$, and for
    all accepted $i$, $\Distribution_{\tilde{x}}(i) \geq \gamma$. This lets us bound the
    bias of the estimator as
    \begin{align}
        \abs{ \Expectation_{i \sim \TXDist} \widehat{x^\dagger y}(i) - x^\dagger y \mkern 3mu} & \leq \abs{ \mkern10mu \smashoperator{\sum_{i:{\TXDist(i) \leq 2\gamma}}} x(i)^* y(i) \mkern 5mu } + \abs{ \mkern 10mu \smashoperator{\sum_{i:{\TXDist(i) \geq \gamma}}} \mkern 5mu \frac{\TXDist(i)}{\hat{p}_i} \hat{x}_i^* y(i) - x(i)^* y(i) \mkern 5mu }                 \\
                                                                                               & \leq \sqrt{2\gamma} \norm{\tilde{x}} \norm{y}_1 + \smashoperator{\sum_{i:{\TXDist(i) \geq \gamma}}} \mkern 3mu \abs{y(i)} \{\abs{x(i)}\,\abs{\frac{\TXDist(i)}{\hat{p}_i + \gamma/8} - 1} + \abs{\hat{x}_i - x(i)}[1 + \abs{\frac{\TXDist(i)}{\hat{p}_i + \gamma/8} - 1}]\} \\
                                                                                               & \leq \frac{1}{4} \epsilon \norm{y}_1 + \smashoperator{\sum_{i:{\TXDist(i)\geq\gamma}}} \mkern 3mu \abs{y(i)} [\norm{\tilde{x}}\frac{\gamma}{4\sqrt{\TXDist(i)}} + \frac{5}{16}\epsilon]                                                                                     \\
                                                                                               & \leq \frac{2}{3} \epsilon \norm{y}_1.
    \end{align}
    The first inequality follows from the definition of the estimator, the triangle
    inequality, and possibly double counting, the second inequality from observations (b)
    and (c), and the third inequality from the definition of $\gamma$ and all of the
    preliminary observations.

    The variance is bounded as
    \begin{align} \hskip-2cm
        \Variance_{i \sim \Distribution_{\tilde{x}}} \widehat{x^\dagger y}(i) & \leq \Expectation_{i\sim\Distribution_{\tilde{x}}} \abs{\widehat{x^\dagger y}(i)}^2                                                                                                                                                                                                                          \\
                                                                              & \leq \sum_{i: {\hat{p}_i \geq \frac{3}{2}\gamma}} \frac{\Distribution_{\tilde{x}}(i)}{\hat{p}_i + \gamma/8} \frac{1}{\hat{p}_i + \gamma/8} \abs{\hat{x}_i}^2 \abs{y(i)}^2                                                                                                                                    \\
                                                                              & \leq \sum_{i: {\hat{p}_i \geq \frac{3}{2}\gamma}} \frac{1}{\Distribution_{\tilde{x}}(i)}[\abs{x(i)} + \frac{\epsilon}{4}]^2 \abs{y(i)}^2                                                                                                                                                                     \\
                                                                              & = \sum_{i: {\hat{p}_i \geq \frac{3}{2}\gamma}} \frac{\abs{x(i)}^2}{\Distribution_{\tilde{x}}(i)} \abs{y(i)}^2 + \frac{(\epsilon/4)^2}{\Distribution_{\tilde{x}}(i)}\abs{y(i)}^2 + \frac{\epsilon/2}{\sqrt{\Distribution_{\tilde{x}}(i)}} \frac{\abs{x(i)}}{\sqrt{\Distribution_{\tilde{x}}(i)}} \abs{y(i)}^2 \\
                                                                              & \leq \norm{\tilde{x}}^2 \norm{y}^2 + \frac{3}{2} 4 \hat{n} \norm{\tilde{x}} \norm{y}^2 + 2 \hat{n} \norm{\tilde{x}} \norm{y}^2                                                                                                                                                                               \\
                                                                              & \leq 10 \norm{\tilde{x}}^2 \norm{y}^2.
    \end{align}
    The first inequality follows from the definition of variance for complex random
    variables, the second inequality from a triangle inequality (and the definition of
    the estimator), and the third inequality from the observation that $\hat{p}_i \geq
    \Distribution_{\tilde{x}}(i) - \gamma/2$. The second to last inequality follows from
    the definition of $\Distribution_{\tilde{x}}(i)$ as
    $\abs{\tilde{x}(i)}^2/\norm{\tilde{x}}^2$ and the imposed condition that
    $\abs{\tilde{x}(i)} \geq \abs{x(i)}$, for all $i$, as well as the summing index
    condition and the definition of $\gamma$.

    With these bounds, it follows from the Chebyshev inequality that the average of
    $\BigO(\frac{\norm{\tilde{x}}^2 \norm{y}^2}{\epsilon^2 \norm{y}^2}) =
    \BigO(\hat{n}^2\epsilon^{-2})$ valid samples of this estimator suffice to estimate a
    value that is at most $\frac{1}{3}\epsilon\norm{y}$ away from the estimator's mean
    --- thus, at most $\epsilon\norm{y}_1$ away from $x^\dagger y$ --- with probability
    at least $1/18$.

    From a Hoeffding inequality, taking the median of $\BigO[\log(\hat{n}^2
    \epsilon^{-2})]$ independent estimates ensures that every $\hat{x}_i$ is a
    correct estimate of $x(i)$ to error $\epsilon/4$, with overall error probability
    $1/18$. From a union bound, it follows that the whole procedure estimates $x^\dagger
    y$ to error $\epsilon\norm{y}_1$ with probability $2/3$.

    To find the stated runtime, we first note that with very high probability, sampling
    $k$ valid samples from $\phi$-oversampling access to $x$ takes time
    $\BigO(k\Sample)$. This follows from a Hoeffding inequality, the one-sidedness of the
    error, and the fact that, by definition, the probability of a valid sample is at
    least $2/3$. Then it follows immediately that the total runtime of the procedure
    outlined above and in algorithm~\ref{alg:inprod-asym} is
    \begin{equation}
        \BigO[\QNorm(\norm{\tilde{x}}^2/4)(1 + \log\frac{1}{\norm{\tilde{x}}^2})] + \BigO(\Sample \gamma^{-2} \log d) + \BigO(\Sample \hat{n}^2 \epsilon^{-2}) + \BigO\{[\Query(\epsilon/4) \log(\hat{n}^2 \epsilon^{-2}) + \log d] \hat{n}^2 \epsilon^{-2}\}.
    \end{equation}

    Using the definition of $\gamma$, the fact that $\hat{n}^2 \leq
    \frac{5}{4}\norm{\tilde{x}}^2$, that $\norm{\tilde{x}}^2 \leq \phi \norm{x}^2$ by
    definition, and the assumption that constant factors can be ``moved out'' of the cost
    runtime functions, the stated runtime follows.
\end{proof}

\medglue

\begin{remark} \label{remark:inprod-asymm-approx}
    A restricted form of this lemma extends to the case where sampling is not exact in the
    distribution. In particular: if $x \in \Complex^d$ such that $\norm{x}\leq 1$, and
    $x' \in \Complex^d$ such that
    \begin{equation}
        \sum_{i\in[d]} \abs{\Distribution_{\tilde{x}}(i) - \Distribution_{x'}(i)} \leq \frac{C\epsilon}{\sqrt{\phi}}
    \end{equation}
    for $\Distribution_{\tilde{x}}$ a distributed $\phi$-oversampling $\Distribution_x$,
    and for sufficient choice of constant $C$, then $\phi$-Approximate Sample access to
    $x'$ (Def.~\ref{def:as}) and Approximate Query access to $x$ (Def.~\ref{def:aq})
    suffices to estimate $x^\dagger y$ to additive error $\epsilon\norm{y}_1$, with
    runtime
    \begin{equation}
        \tilde\BigO[\phi\epsilon^{-2}\log d + \phi\epsilon^{-2}\Query(\epsilon) + \phi^2\epsilon^{-4} \log d \Sample]
    \end{equation}
    with very high probability and under the same assumptions as those in the lemma
    above. This requires that $\phi$ or an upper bound thereof is known (whereas
    Lemma~\ref{lemma:inprod-asym} does not). The proof proceeds like that of
    Lemma~\ref{lemma:inprod-asym} with technical complications, and so is deferred to
    Appendix~\ref{app:inprod-asymm-approx}.
\end{remark}

\bigglue

While the dependence on the $1$-norm of $y$ provides for a weaker result than that of,
for example, Ref.~\cite{Tang2019}, where it is possible to attain absolute error in time
completely independent of the dimension of the underlying vectors, we will nonetheless
find that such a dependence is sufficient to prove, for example,
Theorem~\ref{thm:pauli-inprod}. On the other hand, this $1$-norm has a simple geometric
interpretation: the more $y$ is ``aligned'' with a single measurement basis, the easier
it is to estimate the inner product. At the same time, we note that there is no
dependence on the $1$-norm of $x$, such that only the $1$-norm of the classically queried
vector plays a role in increasing the runtime. This means that even if $x$ has maximal
$1$-norm (e.g., $\abs{x(i)} = d^{-1/2}\norm{x}$), the procedure is still efficient
provided that $y$ has a small $1$-norm (for example, that $y$ is known to be sufficiently
sparse, even if which elements are nonzero is unknown).

Now, we analyze the symmetric case, in which one has \textit{approximate sample and query} access to two
vectors. In light of the remarks of section~\ref{sec:satsf-asq}, this case can be taken
to translate a restricted view of the case where one knows how to prepare and measure two
states, and wishes to estimate the inner product of the underlying state vectors. We find
that if, per the above Theorem, $1$-norm reflects ``peakedness'' in the sampling, only one of the provided vectors needs to be sufficiently ``peaked'' in
the symmetric case.

\begin{figure*}
    \begin{algorithm}[H]
        \caption{Calculating an inner product estimate via $\ASQ[\phi]$ access to both
            vectors ($x$ and $y$), with success probability at least $2/3$ and absolute
            error at most $\epsilon [1 + \min(\frac{\norm{x}_1}{\norm{x}},
                    \frac{\norm{y}_1}{\norm{y}})]$.}
        \label{alg:inprod-sym}

        \algblockdefx{DoTimes}{EndDoTimes}[1]{\textbf{do} #1 \textbf{times:}}{\textbf{end}}

        \begin{varwidth}[t]{.5\textwidth}
            \begin{algorithmic}[1]
                \Require \ASQ[\phi] access to $x, y \in \Complex^d$
                \Require Precision $\epsilon \in (0, 1]$
                \Statex
                \Function{Improve}{$x$, $\epsilon$, $\delta$}
                \State $r \gets []$
                \State $c \gets []$
                \State $\epsilon \gets \epsilon/\sqrt{2}$
                \DoTimes{$\lceil 18\log(1/\delta)\rceil$}
                \State Sorted insert $\Re[\hat{x}(\epsilon)]$ into $r$
                \State Sorted insert $\Im[\hat{x}(\epsilon)]$ into $c$
                \EndDoTimes
                \State\Return $\textsc{median}(r) + \I\, \textsc{median}(c)$
                \EndFunction
                \Statex
                \Function{Sample}{}
                \State $z\gets\textsc{fair coin}$
                \If{$z$ is \textsc{heads}}
                \Repeat \; $i \gets \Call{oversample}{x}$
                \Until{sample is valid.}
                \Else
                \Repeat \; $i \gets \Call{oversample}{y}$
                \Until{sample is valid.}
                \EndIf
                \State\Return $i$
                \EndFunction
                \algstore{alg:inprod-sym}
            \end{algorithmic}
        \end{varwidth}
        \hfil
        \begin{varwidth}[t]{.5\textwidth}
            \begin{algorithmic}[1]
                \algrestore{alg:inprod-sym}
                \State $\hat{n}_x^2 \gets \Call{[algorithm~\ref{alg:rel-err-estim}]}{\norm{x}^2, \rho=\frac{1}{2}, \delta=\frac{1}{18}}$
                \State $\hat{n}_y^2 \gets \Call{[algorithm~\ref{alg:rel-err-estim}]}{\norm{y}^2, \rho=\frac{1}{2}, \delta=\frac{1}{18}}$
                \Statex
                \State $\gamma \gets \min(1, \hat{n}_x^{-2}) \min(1, \hat{n}_y^{-2}) \epsilon^2/100$
                \State $\epsilon_x \gets \epsilon \sqrt{\gamma} \min(1, \hat{n}_y^{-1}) / 100$
                \State $\epsilon_y \gets \epsilon \sqrt{\gamma} \min(1, \hat{n}_x^{-1}) / 100$
                \Statex
                \State $l \gets [\textsc{empty histogram}]$
                \DoTimes{$32 \gamma^{-2} \log(18 d)$}
                \State $i \gets \Call{sample}$
                \State Insert $i$ into $h$.
                \EndDoTimes
                \Statex
                \State $m \gets 864 \, (1 + 2\hat{n}_x^2) (1 + 2\hat{n}_y^2) \, \epsilon^{-2}$
                \State $l \gets []$
                \DoTimes{$m$}
                \State $i \gets \Call{sample}$
                \State $\hat{p}_i \gets \Call{frequency}{h; i}$
                \If{$\hat{p}_i \leq \frac{3}{2}\gamma$}
                \State $v \gets 0$
                \Else
                \State $\hat{x}_i \gets \Call{improve}{x(i), \epsilon_x, 1/(18m)}$
                \State $\hat{y}_i \gets \Call{improve}{y(i), \epsilon_y, 1/(18m)}$
                \State $v \gets (\hat{p}_i + \gamma/2)^{-1} \hat{x}_i^* \hat{y}_i$
                \EndIf
                \State Sorted insert $v$ into $l$
                \EndDoTimes
                \State \Return $\Call{median}{l}$
            \end{algorithmic}
        \end{varwidth}
        \vskip 2pt
    \end{algorithm}
\end{figure*}

\begin{lemma} \TheoremName{Symmetric inner product estimation}
    \label{lemma:inprod-sym}
    Let $x, y \in \Complex^d$, and $\epsilon \in (0, 1]$. Suppose $\ASQ[\phi]$ access to
    $x$ and $y$ is given. One may estimate $x^\dagger y$ to absolute error $\epsilon[1 +
    \min(\norm{x}_1/\norm{x}, \norm{y}_1/\norm{y})]$ with probability $2/3$ by
    algorithm~\ref{alg:inprod-sym}. Assuming, for simplicity, that $\norm{x},\norm{y}\leq
    1$, the runtime of the algorithm is
    $$ \tilde\BigO[\phi^2 \epsilon^{-2} \log d + \phi^2 \epsilon^{-2} \Query(\epsilon^2 \norm{x}\norm{y}) + \epsilon^{-4} \log d \Sample + \QNorm(\norm{x}^2) + \QNorm(\norm{y}^2)] $$
    with very high probability, and under the assumption that $\Query[\Omega(\epsilon)] =
    \BigO[\Query(\epsilon)]$, and likewise for $\QNorm$.
\end{lemma}
\begin{proof}
    The proof proceeds similarly to that of Lemma~\ref{lemma:inprod-asym}.

    Denote the oversampling distribution of $\ASQ[\phi](x)$ by
    $\Distribution_{\tilde{x}}$ and that of $\ASQ[\phi](y)$ by
    $\Distribution_{\tilde{y}}$, for ``oversampling vectors'' $\tilde{x}$ and $\tilde{y}$
    for $x$ and $y$, respectively. Let $\Distribution$ be the distribution obtained by
    uniformly randomly outputting from either $\Distribution_{\tilde{x}}$ or
    $\Distribution_{\tilde{y}}$, from which we have the probability mass function
    $\Distribution(i) = [\Distribution_{\tilde{x}}(i) + \Distribution_{\tilde{y}}(i)]/2$,
    for all $i\in[d]$. Note that $\Distribution$ is a $2$-oversampling distribution for
    both $\Distribution_{\tilde{x}}$ and $\Distribution_{\tilde{y}}$. Thus also, for all
    $i\in[d]$, $\TXDist(i) = \abs{\tilde{x}(i)}^2/\norm{\tilde{x}}^2 \leq
    2\Distribution(i)$, implying $\abs{x(i)} \leq \abs{\tilde{x}(i)} \leq
    \sqrt{2\Distribution(i)}\norm{\tilde{x}}$. The analogous statements for $y$ also
    hold.

    Let $\hat{n}_x^2$ and $\hat{n}_y^2$ be relative error $1/2$ estimates of
    $\norm{\tilde{x}}^2$ and $\norm{\tilde{y}}^2$, respectively, with joint probability
    $8/9$. These are obtained by algorithm~\ref{alg:rel-err-estim} and the $\ASQ[\phi]$
    access, and so with high probability in time $\BigO\{\QNorm(\norm{\tilde{x}^2})[1 +
    \log(1/\norm{\tilde{x}}^2)]\}$ for $\hat{n}_x^2$, and
    $\BigO\{\QNorm(\norm{\tilde{y}^2})[1 + \log(1/\norm{\tilde{y}}^2)]\}$ for
    $\hat{n}_y^2$. Note that, with these choices, $1/2 \norm{\tilde{x}}^2 \leq
    \hat{n}_x^2 \leq 3/2 \norm{\tilde{x}}^2$, from which follows
    $\hat{n}_x/\norm{\tilde{x}} < \sqrt{2}$, and likewise for $\hat{n}_y, y$.

    Define $\gamma \Defined \Theta[\min(1, \hat{n}_x^2) \min(1, \hat{n}_y^2)
    \epsilon^2]$. From Theorem~\ref{thm:smpl-estim}, $\BigO(\gamma^{-2} \log d)$ samples
    of $\Distribution$ are sufficient to know, for every $i=1,\ldots,d$,
    $\Distribution(i)$ to error $\gamma/2$ and with overall probability $8/9$. [Looking
    up each estimate of $\Distribution(i)$, say $\hat{p}_i$, takes time $\BigO(\log d)$,
    per the Theorem.]

    Furthermore, for $i \sim \Distribution$, consider $\hat{x}_i$ and $\hat{y}_i$ to be
    estimates of, respectively, $x(i)$ and $y(i)$, to absolute error, respectively,
    $\Theta[\sqrt{\gamma} \epsilon / \max(1,\hat{n}_y)]$ and $\Theta[\sqrt{\gamma}
    \epsilon / \max(1,\hat{n}_x)]$. With these, we define the point estimator
    \begin{equation}
        \widehat{x^\dagger y}(i) = \begin{cases}
            0                                                    & \text{if } \hat{p}_i \leq \frac{3}{2} \gamma \\
            \frac{1}{\hat{p}_i + \gamma/2} \hat{x}_i^* \hat{y}_i & \text{otherwise.}
        \end{cases}
    \end{equation}

    This estimator has bias (with respect to $x^\dagger y$) bounded as (from a triangle
    inequality and double counting)
    \begin{align}
        \abs{\Expectation_{i\sim\Distribution} \widehat{x^\dagger y} - x^\dagger y} & \leq \sum_{i:\Distribution(i)\leq2\gamma} \abs{x(i)}\abs{y(i)} \mkern 5mu + \mkern -5mu \sum_{i:\Distribution(i)\geq\gamma} \abs{\frac{\Distribution(i)}{\hat{p}_i + \gamma/2} \hat{x}_i^* \hat{y}_i - x(i)^* y(i)}.
    \end{align}
    The first term is further bounded as
    \begin{multline}
        \sum_{i:\Distribution(i) \leq 2\gamma} \abs{x(i)} \abs{y(i)} \leq 2\sqrt{\gamma} \min(\norm{\tilde{x}}\abs{y(i)}, \abs{x(i)}\norm{\tilde{y}})
        = \frac{1}{5} \frac{\norm{\tilde{x}}\norm{\tilde{y}}}{\hat{n}_x \hat{n}_y} \epsilon \min(\frac{\norm{x}_1}{\norm{x}}, \frac{\norm{y}_1}{\norm{y}}) \\
        \leq \frac{2}{5} \epsilon \min(\frac{\norm{x}_1}{\norm{x}}, \frac{\norm{y}_1}{\norm{y}})
    \end{multline}
    where we have used $\abs{x(i)} \leq \sqrt{2\Distribution(i)}$ and the condition on
    the summing index. The second term is further bounded as
    \begin{multline}
        \sum_{i:\Distribution(i)\geq\gamma} \abs{\frac{\Distribution(i)}{\hat{p}_i + \gamma/2} \hat{x}_i^* \hat{y}_i - x(i)^* y(i)} \\ \leq \sum_{i:\Distribution(i)\geq\gamma} \abs{\frac{\Distribution(i)}{\hat{p}_i + \gamma/2} - 1} \abs{x(i)}\abs{y(i)} + \abs{\hat{x}_i^* \hat{y}_i - x(i)^* y(i)} + \abs{\frac{\Distribution(i)}{\hat{p}_i + \gamma/2} - 1}\abs{\hat{x}_i^* \hat{y}_i - x(i)^* y(i)}.
    \end{multline}
    which can then be further bounded term-by-term:
    \begin{flalign}
        \sum_{i:\Distribution(i)\geq\gamma} \abs{\frac{\Distribution(i)}{\hat{p}_i + \gamma/2}} \abs{x(i)}\abs{y(i)} & \leq \sum_{i:\Distribution(i)\geq\gamma} \frac{\gamma}{\Distribution(i)}\abs{x(i)}\abs{y(i)}                                                                                                                                                                   & \\
                                                                                                                     & \leq \frac{\sqrt{2}}{10} \frac{\norm{\tilde{x}}\norm{\tilde{y}}}{\hat{n}_x \hat{n}_y} \epsilon \min(\frac{\norm{x}_1}{\norm{x}}, \frac{\norm{y}_1}{\norm{y}}) \leq \frac{\sqrt{2}}{5} \epsilon \min(\frac{\norm{x}_1}{\norm{x}}, \frac{\norm{y}_1}{\norm{y}}); &
    \end{flalign}
    \begin{flalign}
        \sum_{i:\Distribution(i)\geq\gamma} \abs{\hat{x}_i^* \hat{y}_i - x(i)^* y(i)} & \leq \sum_{i:\Distribution(i)\geq\gamma} \abs{\hat{x}_i - x(i)}\abs{\hat{y}_i - y(i)} + \abs{x(i)} \abs{\hat{y}_i - y(i)} + \abs{\hat{x}_i - x(i)}\abs{y(i)}                                                                                              & \\
                                                                                      & \leq \sum_{i:\Distribution(i)\geq\gamma} (\frac{1}{100}\epsilon)^2 \gamma + \frac{\sqrt{2}}{100} \sqrt{\Distribution(i) \gamma} \, [\frac{\norm{\tilde{x}}}{\max(1, \hat{n}_x)} + \frac{\norm{\tilde{y}}}{\max(1, \hat{n}_y)}] \leq \frac{1}{4} \epsilon; &
    \end{flalign}
    \begin{multline}
        \sum_{i:\Distribution(i)\geq\gamma} \abs{\frac{\Distribution(i)}{\hat{p}_i + \gamma/2} - 1}\abs{\hat{x}_i^* \hat{y}_i - x(i)^* y(i)} \leq\sum_{i:\Distribution(i)\geq\gamma} \frac{\gamma}{\Distribution(i)}\abs{\hat{x}_i^* \hat{y}_i - x(i)^* y(i)} \\
                                                                                                                                             \leq \sum_{i:\Distribution(i)\geq\gamma} \abs{\hat{x}_i^* \hat{y}_i - x(i)^* y(i)} \mkern-6mu\mkern12mu\leq \frac{1}{4}\epsilon.
    \end{multline}
    Above, we've used the aforementioned fact that $\abs{\hat{a} - a} \leq \epsilon$
    implies $\abs{a/\hat{a} - 1} \leq 2\epsilon/a$, the relation between $\abs{x(i)}$ and
    $\Distribution(i)$ (and between $\abs{y(i)}$ and $\Distribution(i)$), the condition
    on the summing index (to upper bound $\gamma$ by $\Distribution(i)$ where necessary),
    and the various error bounds.

    Thus we conclude that the bias of the estimator is bounded as $\BigO\{ \epsilon[1 +
    \min(\frac{\norm{x}_1}{\norm{x}}, \frac{\norm{y}_1}{\norm{y}})] \}$.

    Likewise, we bound the variance of the estimator as
    \begin{align}
        \Variance_{i\sim\Distribution} \widehat{x^\dagger y} & \leq \Expectation_{i\sim\Distribution} \abs{\widehat{x^\dagger y}}^2 \leq \sum_{i:\Distribution(i)\geq\gamma} \frac{1}{\Distribution(i)} [\abs{x(i)} + \frac{1}{10} \gamma \max(1, \hat{n}_x)]^2 [\abs{y(i)} + \frac{1}{10} \gamma \max(1, \hat{n}_y)]^2 \\
                                                             & \leq \sum_{i:\Distribution(i)\geq\gamma} \frac{4}{\Distribution(i)} [\abs{x(i)}^2 + \frac{1}{100} \gamma^2 \max(1, \hat{n}_x^2)] [\abs{y(i)}^2 + \frac{1}{100} \gamma^2 \max(1, \hat{n}_y^2)]                                                            \\
                                                             & \leq 10 (1 + \norm{\tilde{x}}^2)(1 + \norm{\tilde{y}}^2).
    \end{align}
    The first inequality uses the observations that the error bounds for $\hat{x}_i$ may
    be written as $\gamma \max(1, \hat{n}_x)/10$ (similarly for $\hat{y}_i$), and that
    the error bounds on $\hat{p}_i$ imply that $(\hat{p}_i + \gamma/2)^{-1} \leq
    \Distribution(i)^{-1}$. The second inequality uses a geometric mean--quadratic mean
    inequality. The third inequality uses a Cauchy-Schwarz inequality, the summing
    index's condition, and the inequality $\norm{x}^2 \leq \norm{\tilde{x}}^2$ (likewise
    for $y$).

    Then it follows from a Chebyshev bound that $\BigO[\epsilon^{-2} (1 +
    \norm{\tilde{x}}^2) (1 + \norm{\tilde{y}}^2)] = \BigO[\epsilon^{-2} (1 + \hat{n}_x^2)
    (1 + \hat{n}_y^2)]$ valid samples have their average at most $\epsilon/2$ away from
    the estimator's mean --- thus are at most $\epsilon [1 +
    \max(\frac{\norm{x}_1}{\norm{x}}, \frac{\norm{y}_1}{\norm{y}})]$ away from $x^\dagger
    y$ --- with probability $1/18$. From a Hoeffding bound, the error probability of each
    sample of $\hat{x}_i$ and $\hat{y}_i$ can be reduced to ensure that every sample is a
    correct estimate of $x(i)$, $y(i)$, respectively, with overall probability $1/18$.
    From a union bound, it follows that the final estimation is correct with probability
    $2/3$.

    From the outlined procedure, and explicitly from algorithm~\ref{alg:inprod-sym}, one
    finds that the runtime is, with high probability,
    \begin{multline}
        \BigO[\QNorm(\norm{\tilde{x}}^2/2)(1 + \log\frac{1}{\norm{\tilde{x}}^2})] + \BigO[\QNorm(\norm{\tilde{y}}^2/2)(1 + \log\frac{1}{\norm{\tilde{y}}^2})] \\+ \BigO(\Sample \gamma^{-2} \log d) + \BigO[\epsilon^{-2} (1 + \norm{\tilde{x}}^2)(1 + \norm{\tilde{y}}^2) \log d] \\+ \BigO\{\epsilon^{-2} (1 + \norm{\tilde{x}}^2)(1 + \norm{\tilde{y}}^2) \log[\epsilon^{-2}(1 + \norm{\tilde{x}}^2)(1 + \norm{\tilde{y}}^2)] \Query[\epsilon^2 \min(1, \norm{\tilde{x}}) \min(1, \norm{\tilde{y}})]\}.
    \end{multline}

    The stated complexity follows.
\end{proof}

\medglue

\begin{remark} \label{remark:inprod-sym-approx}
    A restricted form of this lemma extends to the case where sampling is not exact in
    either distribution. In particular: let $\phi \geq 1$, $\epsilon \in (0, 1]$, and $x,
        y \in \Complex^d$, such that $\norm{x},\norm{y} \leq 1$. Suppose $x', y' \in
        \Complex^d$ satisfy
    \begin{alignat}{2}
        \sum_{i\in[d]} \abs{\Distribution_{\tilde{x}}(i) - \Distribution_{x'}(i)} & \leq \frac{C\epsilon}{\phi}, & \qquad  \sum_{i\in[d]} \abs{\Distribution_{\tilde{y}}(i) - \Distribution_{y'}(i)} & \leq \frac{C\epsilon}{\phi}
    \end{alignat}
    where $\Distribution_{\tilde{x}}$ is a distribution $\phi$-oversampling
    $\Distribution_x$, $\Distribution_{\tilde{y}}$ is a distribution $\phi$-oversampling
    $\Distribution_y$, and $C \leq 1$ is a suitably chosen constant. Then, Approximate
    Sample access to $x'$ and $y'$ and Approximate Query access to $x$ and $y$ suffices
    to estimate $x^\dagger y$ to additive error $\epsilon[1 + \min(\norm{x}_1,
    \norm{y}_1)]$ and success probability $2/3$, halting with high probability in time
    $$\BigO[\phi^2 \epsilon^{-2} \log d + \phi^2 \epsilon^{-2} \Query(\norm{x}\norm{y}\epsilon^2) + \phi^2 \epsilon^{-4} \log d \Sample].$$

    This requires that $\phi$ or an upper bound to $\phi$ is known (whereas the above
    lemma does not). The proof proceeds similarly to that of
    Lemma~\ref{lemma:inprod-sym}, and so is deferred to
    Appendix~\ref{app:inprod-sym-approx}.
\end{remark}

\bigglue

The lemmas above have assumed the generality of complex vectors and oversampling access,
such that if $\kappa$ is an upper bound to $\min(\norm{x}_1, \norm{y}_1)$ is known, and
one wishes to attain absolute error on the estimate of $x^\dagger y$, the runtime scales
with at least $\kappa^4$. If one is willing to restrict the considered setting, and if
querying for entries of $x$ and $y$ is not much more expensive than sampling, then this
dependence on $\kappa$ can be improved:

\begin{lemma} \label{lemma:inprod-real}
    Let $x, y \in \Real^d$ such that $\norm{x} = \norm{y} = 1$. Choose $\epsilon \in
        (0,1]$ and $\Delta \in (0, 1]$. Assume Approximate Query access to $x$ and $y$ and
    Approximate Sample access to $x'$ and $y'$, where $\norm{x - x'} \leq \Delta/2$ and
    $\norm{y - y'} \leq \Delta/2$. Let $\kappa$ be a known upper bound to
    $\min(\norm{x}_1, \norm{y}_1)$. Then, $x^T y$ may be estimated to an absolute error of
    $\epsilon + \Delta$ with probability $2/3$ in time
    $$ \tilde\BigO[\epsilon^{-2} \Sample + \Query(\epsilon^2 \kappa^{-1})]. $$
\end{lemma}
\begin{proof}
    Observe that, for fixed $i \in [d]$, at cost $\tilde\BigO[\Query(\epsilon)]$, both
    $x(i)$ and $y(i)$ can be known, with high probability, to additive error $\epsilon$.
    Define $\Distribution$ as that resulting from uniformly randomly sampling from either
    $\Distribution_x$ or $\Distribution_y$. If $\hat{x}_i$, $\hat{y}_i$ are the
    estimates for, respectively, $x(i)$ and $y(i)$, then by a geometric argument one has
    that
    \begin{equation} \label{eq:sqrt-p-estim}
        \sqrt{\hat{p}_i} \Defined \sqrt{\frac{\hat{x}_i^2 + \hat{y}_i^2}{2}} \implies \abs{\sqrt{\hat{p}_i} - \sqrt{\Distribution(i)}} \leq \epsilon.
    \end{equation}

    Define, similarly, $\AltDistribution$ as the distribution corresponding to uniformly
    randomly sampling from either $\Distribution_{x'}$ or $\Distribution_{y'}$. From
    Lemma~\ref{lemma:approx-tvd}, one has that
    \begin{equation}
        \sum_{i\in[d]} \abs{\AltDistribution(i) - \Distribution(i)} \leq 2\Delta.
    \end{equation}

    Let $C_1$ and $C_2$ be suitably chosen constants, and define $\gamma \Defined
    C_1\epsilon/\kappa$ and $\epsilon_Q \Defined C_2\epsilon^2/\kappa$. By
    eq.~\eqref{eq:sqrt-p-estim}, take $\sqrt{\hat{p}_i}$ and $\hat{p}_i$ be estimates of,
    respectively, $\sqrt{\Distribution(i)}$ and $\Distribution(i)$ to error $\gamma/2$.
    Furthermore let $\hat{x}_i$ and $\hat{y}_i$ be estimates of, respectively, $x(i)$ and
    $y(i)$ to error $\epsilon_Q$. Obtaining these estimates (with high probability) takes
    time $\BigO[\Query(\gamma) + \Query(\epsilon_Q)]$, by the Approximate Query access.

    Now, take a point estimator like that of lemmas~\ref{lemma:inprod-asym} and
    \ref{lemma:inprod-sym}, but ``rejecting'' samples based on the value of
    $\sqrt{\hat{p}_i}$:
    \begin{equation}
        \widehat{x^T y}(i) \leq \begin{cases}
            0                                     & \text{if } \sqrt{\hat{p}_i} \leq \frac{3}{2} \gamma \\
            \frac{\hat{x}_i \hat{y}_i}{\hat{p}_i} & \text{otherwise.}
        \end{cases}
    \end{equation}

    We start by bounding the bias of this estimator under sampling of $\Distribution$,
    rather than under sampling of $\AltDistribution$:
    \begin{equation} \label{eq:bias-under-p}
        \abs{\Expectation_{i\sim\Distribution} \widehat{x^T y}(i) - x^T y} \mkern 10mu \leq \mkern -10mu \sum_{i: \Distribution(i)^{1/2}\leq 2\gamma} \abs{x(i)} \abs{y(i)} \mkern 8mu + \mkern -8mu \sum_{i: \Distribution(i)^{1/2}\geq\gamma} \abs{\Distribution(i) \frac{\hat{x}_i \hat{y}_i}{\frac{\hat{x}_i^2 + \hat{y}_i^2}{2}} - x(i)y(i)}.
    \end{equation}
    Since, from its definition, $2\Distribution(i) \geq x(i)^2$ and $2\Distribution(i)
    \geq y(i)^2$, and from the condition in the summation, it is straightforward to see
    that
    \begin{equation}
        \sum_{i: \Distribution(i)^{1/2}\leq 2\gamma} \abs{x(i)} \abs{y(i)} \leq 2\sqrt{2} \gamma \min(\norm{x}_1, \norm{y}_1) = 2\sqrt{2} C_1 \epsilon.
    \end{equation}
    To bound the right-hand side term of eq.~\eqref{eq:bias-under-p}, we use the fact
    that $\hat{p}_i$ is always ``consistent'' with $\hat{x}_i$ and $\hat{y}_i$, and that
    the sampling is ``exact'' (in the sense that we are not oversampling). Namely, we
    have that
    \begin{align}
        \sum_{i: \Distribution(i)^{1/2}\geq\gamma} \abs{\Distribution(i) \frac{\hat{x}_i \hat{y}_i}{\frac{\hat{x}_i^2 + \hat{y}_i^2}{2}} - x(i)y(i)} \mkern 5mu = \mkern 5mu \smashoperator{\sum_{i: \Distribution(i)^{1/2}\geq\gamma}} \mkern 5mu 2\Distribution(i) \, \abs{f(\hat{x}_i, \hat{y}_i) - f[x(i), y(i)]},
    \end{align}
    where $f(a,b) \Defined \frac{ab}{a^2 + b^2}$. On the other hand, for $\sqrt{a^2 +
    b^2} \geq \sqrt{2}\gamma$, $f$ is Lipchitz continuous, with constant
    $(\sqrt{2}\gamma)^{-1}$, as may be seen:
    \begin{gather}
        \abs*{\frac{\partial f}{\partial a}}(a,b) = \abs{a} \frac{\abs{a^2 - b^2}}{(a^2 + b^2)^2} \leq \frac{\abs{a}}{a^2 + b^2}, \quad
        \abs*{\frac{\partial f}{\partial b}}(a,b) = \abs{b} \frac{\abs{a^2 - b^2}}{(a^2 + b^2)^2} \leq \frac{\abs{b}}{a^2 + b^2}; \\
        \implies \norm{\nabla f} \leq \frac{\sqrt{a^2 + b^2}}{a^2 + b^2} \leq \frac{1}{\sqrt{2}\gamma}.
    \end{gather}

    Moreover, $\norm{(\hat{x}_i, \hat{y}_i) - (x(i), y(i))} \leq \sqrt{2}\epsilon_Q$.
    Thus it follows that
    \begin{align}
        \sum_{i: \Distribution(i)^{1/2}\geq\gamma} \abs{\Distribution(i) \frac{\hat{x}_i \hat{y}_i}{\frac{\hat{x}_i^2 + \hat{y}_i^2}{2}} - x(i)y(i)} \mkern 5mu \leq \mkern -5mu \sum_{i: \Distribution(i)^{1/2}\geq\gamma} \mkern -5mu 2\Distribution(i) \frac{1}{\sqrt{2}\gamma} \sqrt{2}\epsilon_Q \leq 2\frac{C_2}{C_1} \epsilon.
    \end{align}

    The variance of the estimator under $\Distribution$ is also easily bounded by a
    geometric mean--arithmetic mean inequality:
    \begin{align} \label{eq:var-under-p}
        \Variance_{i \sim \Distribution} \widehat{x^T y}(i) \leq \Expectation_{i\sim\Distribution} \widehat{x^T y}(i)^2 \mkern 5mu = \mkern -13mu \sum_{i: \Distribution(i)^{1/2}\geq\gamma} \mkern -5mu \Distribution(i) \Big( \frac{\hat{x}_i \hat{y}_i}{\frac{\hat{x}_i^2 + \hat{y}_i^2}{2}} \Big)^2 \mkern 5mu \leq \mkern -10mu \sum_{i: \Distribution(i)^{1/2}\geq\gamma} \mkern -8mu \Distribution(i) \leq 1.
    \end{align}

    Now we show that because $\Distribution(i)$ is not estimated by sampling, sampling
    from $\AltDistribution$ does not affect the analysis substantially: the variance is
    still bounded by $1$ by exactly the approach of eq.~\eqref{eq:var-under-p}, and the
    bias is now bounded as
    \begin{equation}
        \abs{\Expectation_{i\sim\AltDistribution} \widehat{x^T y}(i) - x^T y} \leq \abs{\Expectation_{i\sim\AltDistribution} \widehat{x^T y}(i) - \Expectation_{i\sim\Distribution} \widehat{x^T y}(i)} + (2 \sqrt{2}C_1 + 2\frac{C_2}{C_1}) \epsilon \leq (4\sqrt{2} C_1 + 2\frac{C_2}{C_1})\epsilon + \Delta,
    \end{equation}
    since, using again a geometric mean--arithmetic mean inequality,
    \begin{align}
        \abs{\Expectation_{i\sim\AltDistribution} \widehat{x^T y}(i) - \Expectation_{i\sim\Distribution} \widehat{x^T y}(i)} & \mkern 5mu \leq \mkern 5mu \smashoperator{\sum_{i: \Distribution(i)^{1/2}\leq2\gamma}} \mkern 10mu \abs{x(i)}\abs{y(i)} \mkern2mu + \mkern 5mu \smashoperator{\sum_{i: \Distribution(i)^{1/2}\geq\gamma}} \mkern 10mu \abs{\AltDistribution(i) - \Distribution(i)} \abs{\frac{\hat{x}_i \hat{y}_i}{\hat{p}_i}} \\&\leq 2\sqrt{2}\gamma \min(\norm{x}_1, \norm{y}_1) + \frac{1}{2} 2\Delta = 2\sqrt{2}C_1 \epsilon + \Delta.
    \end{align}

    The stated runtime follows from an analysis analogous to that of the proof of
    Lemma~\ref{lemma:inprod-sym}.
\end{proof}

\section{Application to distributed inner product estimation}
\label{sec:dist-inprod}

Finally, we show that the access model encoded by \textit{approximate sample and query}
and the computational results derived in Section~\ref{sec:computation} suffice to derive
a result of independent interest. (We refer the reader to Section~\ref{sec:pauli-samp} and
Ref.~\cite{Hinsche2024} for a review of the concepts considered in the following
statement.)

\begin{theorem} \TheoremName{Distributed inner product estimation under entanglement and stabilizer norm constraints}
    \label{thm:pauli-inprod}
    Let $\ket{\psi}$ and $\ket{\phi}$ be pure states of $\QubitsN$ qubits, such that both have no more than
    $\chi$ bipartite entanglement across any bipartition of their qubits, and have
    stabilizer norm (at most) $\StabNorm$. If preparing $\ket{\psi}$ or $\ket{\phi}$ takes time $\BigO(T)$, and
    assuming Bell measurements on $\ket{\psi}$ or $\ket{\phi}$ take time $\BigO(T)$, then there is an
    (explicit) procedure, which does not require simultaneously coherent copies of $\ket{\psi}$
    and $\ket{\phi}$, and that estimates $\abs{\braket{\psi}{\phi}}^2$ to absolute error $\epsilon$ with
    high probability in time
    \begin{equation}
        \tilde\BigO[\QubitsN^5 e^{4\chi} \StabNorm^2 \epsilon^{-6} + T \epsilon^{-4} \QubitsN^3 e^{4\chi} \StabNorm^2 + T\epsilon^{-2}\StabNorm^2].
    \end{equation}
\end{theorem}

To prove Theorem~\ref{thm:pauli-inprod}, we begin by adapting \cite[Lemmas 2 and
3]{Hinsche2024} to the following:
\begin{lemma} \label{lemma:fadap}
    Let $\rho$ be a $d$-dimensional pure quantum state, and $\tau \in [0,1]$. Denote
    $\alpha_\rho(i) \Defined \Tr(\rho P_i)$, and let $\Distribution_{\pi_\rho}(i) =
        \alpha_\rho(i)^2/d$ be the probability mass function associated to Pauli sampling from
    $\rho$. Let $F_\rho(\tau) \Defined
        \Probability_{i\sim\Distribution_{\pi_\rho}}[\alpha_\rho(i)^2 < \tau]$, and
    $M_{1/2}(\rho)$ be the stabilizer $1/2$-R\'{e}nyi entropy, as defined in
    section~\ref{sec:pauli-samp} (using natural base logarithms). Then,
    \begin{equation}
        F_\rho(\tau) \leq \sqrt{\tau} e^{\frac{1}{2}M_{1/2}(\rho)}.
    \end{equation}
\end{lemma}
\begin{proof}
    We have that, by definition,
    \begin{equation}
        F_\rho(\tau) \mkern3mu = \mkern 8mu \smashoperator{\sum_{i : \alpha_\rho(i)^2 < \tau}} \mkern 8mu \frac{\alpha_\rho(i)^2}{d}.
    \end{equation}

    On the other hand,
    \begin{equation}
        M_{1/2}(\rho) = 2 \log[\smashoperator{\sum_{i\in[d^2]}}\frac{\alpha_\rho(i)}{d}] = 2\log[\smashoperator{\sum_{i\in[d]^2}}\frac{\alpha_\rho(i)^2}{d}\frac{1}{\alpha_\rho(i)}]
    \end{equation}
    implying that
    \begin{equation}
        e^{\frac{1}{2}M_{1/2}(\rho)} = \Expectation_{i\sim\Distribution_{\pi_\rho}}[\frac{1}{\alpha_\rho(i)}].
    \end{equation}

    Let $X$ be the random variable $X = \alpha_\rho(Y)$ for
    $Y\sim\Distribution_{\pi_\rho}$. The above implies that $\Expectation X =
    2^{\frac{1}{2}M_{1/2}(\rho)}$, and so we may apply Markov's inequality:
    \begin{equation}
        \Probability(X \geq \tau^{-1/2}) \leq \frac{e^{\frac{1}{2}M_{1/2}(\rho)}}{\tau^{-1/2}} \quad\implies\quad \Probability[\alpha_\rho(Y)^2 \leq \tau] \leq \sqrt{\tau}e^{\frac{1}{2}M_{1/2}(\rho)}.
    \end{equation}

    But $\Probability_{i\sim\Distribution_{\pi_\rho}}[\alpha_\rho(i) \leq \tau] \equiv
    F_\rho(\tau)$, concluding the proof.
\end{proof}

Then we may directly adapt the proof of \cite[Theorem 8]{Hinsche2024}, making use of the
above lemma:
\begin{corollary} \TheoremName{Corollary of \cite[Theorem 8]{Hinsche2024}}
    \label{corollary:pauli-samp}
    Let $\rho$ be a $d$ dimensional pure quantum state. Let $M_{1/2}(\rho)$ be its
    stabilizer $1/2$-R\'{e}nyi entropy, and $\chi$ be the maximum R\'{e}nyi entanglement
    entropy across all possible bipartitions of $\rho$. Finally, let
    $\Distribution_{\pi_\rho}$ be the distribution with support $\{1, \ldots, d^2\}$
    corresponding to Pauli sampling on $\rho$, as defined in
    section~\ref{sec:pauli-samp}. For a choice of $\Delta,\delta \in (0,1)$, one may
    generate $i\sim\AltDistribution$, where, with probability at least $1-\delta$,
    \begin{equation}
        \sum_{i\in[d^2]}\abs{\Distribution_{\pi_\rho}(i) - \AltDistribution(i)} \leq \Delta,
    \end{equation}
    requiring $N$ copies of $\rho\otimes\rho$ Bell measurements, with
    \begin{equation}
        N = \BigO(e^{4\chi} e^{2M_{1/2}(\rho)} \Delta^{-4} \log^3 d \log \frac{1}{\delta}).
    \end{equation}

    Each generated sample requires $\BigO(N \log^2 d)$ time in classical postprocessing
    of the obtained Bell measurement data.
\end{corollary}
\begin{proof}
    From \cite[Lemma 1]{Hinsche2024} and the adapted Lemma~\ref{lemma:fadap}, and for
    $\mathfrak{f}, \gamma, \epsilon$ as defined in Ref.~\cite{Hinsche2024},
    \begin{equation}
        \mathfrak{f}(\gamma) \leq F_\rho(\gamma e^{2\chi}) \leq \sqrt{\gamma} e^\chi e^{\frac{1}{2}M_{1/2}(\rho)}.
    \end{equation}
    Then, from \cite[Theorem 7]{Hinsche2024}, the algorithm therein proposed samples from
    a distribution $\AltDistribution$ such that
    \begin{equation}
        \sum_{i\in[d^2]} \abs{\Distribution_{\pi_\rho}(i) - \AltDistribution(i)} \leq \mathfrak{f}(\gamma) + \exp(\frac{4\epsilon\log d}{\gamma}) - 1.
    \end{equation}

    Thus we wish to choose $\gamma$ and $\epsilon$ such that
    \begin{equation}
        \left\{\,\begin{alignedat}{2}
             & \mathfrak{f}(\gamma)                 &  & \leq \frac{\Delta}{2}                                 \\
             & \exp(\frac{4\epsilon\log d}{\gamma}) &  & \leq \exp(\frac{\Delta}{4}) \leq 1 + \frac{\Delta}{2}
        \end{alignedat}\right.
        \impliedby
        \left\{\;\begin{alignedat}{1}
             & \gamma \leq (\frac{\Delta}{2})^2 (e^\chi)^{-2} e^{-M_{1/2}(\rho)} \\
             & \epsilon = \frac{\Delta}{16} \frac{\gamma}{\log d}
        \end{alignedat}\right.
    \end{equation}
    ensuring that the total variation distance between $\Distribution_{\pi_\rho}$ and
    $\AltDistribution$ is guaranteed to be at most $\Delta$. The rest of the proof
    proceeds exactly like in Ref.~\cite{Hinsche2024}.
\end{proof}

But now note that, from section~\ref{sec:pauli-samp}, $\abs{\braket{\psi}{\phi}}^2 \equiv
\Tr(\dyad{\psi} \dyad{\phi}) \equiv {\pi_{\dyad{\psi}}}^T \pi_{\dyad{\phi}}$, and
\begin{equation}
    e^{\frac{1}{2}M_{1/2}(\rho)} \equiv \StabNorm(\rho) \equiv d^{-1} \sum_{i \in [d]^2} \abs{\Tr(P_i\rho)} \equiv d^{-1/2}\norm{\pi_\rho}_1.
\end{equation}

Thus Theorem~\ref{thm:pauli-inprod} follows as an immediate corollary of
Remark~\ref{remark:asq-pauli}, Lemma~\ref{lemma:inprod-real}, and
Corollary~\ref{corollary:pauli-samp}.

We note that the approach we have outlined above differs from that of
Ref.~\cite{Hinsche2024} in the way Pauli sampling is used to construct an inner
product estimator. As opposed to their symmetric and antisymmetric protocols, we
leverage \ASQ\space to create an estimator that more directly computes an overlap between the ``Pauli
amplitudes'' (square root of the Pauli distributions of the two states). From
this results the polynomial improvements in sample and time complexity.

We also note that by proving Theorem~\ref{thm:pauli-inprod} with the apparatus
of \textit{approximate sample and query} we provide an alternative
interpretation for why Pauli sampling is useful for the task of distributed
inner product estimation, and why only states with low nonstabilizerness admit
efficient procedures. On one hand, we see that the Pauli representation is
especially cheap to query from when one is allowed to measure the expectation
value of Pauli strings. Due to the normalization factor of $d^{-1/2}$ in each
entry of the Pauli representation of a state, one may learn it to error
$\epsilon/d^{1/2}$ by learning the Pauli string's expectation value to only
error $\epsilon$. On the other hand, we find that the classes of states with low
nonstabilizerness measures correspond to states with low $1$-norm when expressed
in the Pauli basis. Thus, we may argue that for some states sampling in the
Pauli basis is useful because the states' representation in that basis is quite
``peaked''.

\section{Discussion and outlook}

We have shown that the \textit{sample and query} model can be adapted to more accurately
reflect the ability to prepare and measure block encoded states, thereby addressing the
comments of Ref.~\cite{Colter2021}. This adaptation preserves the (arguably) two defining
characteristics of the original \textit{sample and query} model: those of composability, and the
ability to perform meaningful computation with this form of access. However, this
adaptation is not without its drawbacks: the randomized nature of the model poses new
problems in deriving results analogous to those of Ref.~\cite{Chia2020}, and, in general,
the proof techniques employed in Ref.~\cite{Tang2019} and onwards do not seem to
translate well without careful adaptation.

The ultimate goal for such a framework would be to achieve a dequantization result akin
to that of Ref.~\cite{Chia2020} for the ASQ model: not only because of the practical
implications of such a result, but also because of the insight regarding quantum
computational power that this would provide. While we hope to have established a first step in
this direction, it is unclear if such a result should be expected. We can extend
the \textit{approximate sample and query} model to block encodings of matrices --- we do
so in Appendix~\ref{app:asq-bemat} --- but we observe that Ref.~\cite{Chia2020}'s results
strongly rely on the classical technique of ``sketching'' \cite{DKLR1995}. This technique
does not translate immediately to a randomized and finite error setting (as is the case
for the other proof techniques, as stated), and so the next steps in this direction will
require either a careful adaptation or a completely original approach.

Nonetheless, we have shown that the ASQ model is both expressive and satisfiable in varied setups.
Faithfully to the kinds of results obtained with the SQ model,
Theorem~\ref{thm:pauli-inprod} establishes not only ``worst-case scenario'' computational
complexities, but rather a full characterization with respect to the input data. This
provides better bounds for ``well behaved'' data, as well as quantifiable overheads for
``badly behaved'' data.

\acknowledgements

The authors would like to thank M.~Szegedy, D.~Magano, S.~Pratapsi, E.~Trombetti,
L.~Bugalho, M.~Casariego, M.~Hinsche, S.~Jerbi, J.~Carrasco, M.~Ioannou, A.~Kontorovich,
and D.~Berend for their help and comments.
We thank the support from FCT --- Funda\c{c}\~{a}o para a Ci\^{e}ncia e a Tecnologia
(Portugal), namely through project {UIDB/04540/2020}, as well as from projects QuantHEP
and HQCC supported by the EU QuantERA ERA-NET Cofund in Quantum Technologies and by FCT
({QuantERA/0001/2019} and {QuantERA/004/2021}, respectively), from the EU Quantum
Flagship project EuRyQa (101070144), from the BMWi (PlanQK), and the BMBF (QPIC-1, HYBRID). MM would like to further acknowledge and thank FCT
--- Funda\c{c}\~{a}o para a Ci\^{e}ncia e a Tecnologia (Portugal) for its support through
the scholarship with identifier {2021.05528.BD} and DOI {10.54499/2021.05528.BD}.

\bibliographystyle{quantum}
\bibliography{citations}

\appendix

\section{Proof of \texorpdfstring{Remark~\ref{remark:inprod-asymm-approx}}{Remark 27}}
\label{app:inprod-asymm-approx}

Let $\epsilon \in (0, 1]$, $\phi \geq 1$, and $x, x', \tilde{x} \in \Complex^d$, such
that for a suitable constant $C_1 > 0$,
\begin{gather}
    \norm{\tilde{x} - x'} \leq \frac{C_1 \epsilon}{4\sqrt{\phi}} \norm{\tilde{x}}, \\
    \norm{\tilde{x}}^2 \leq \phi \norm{x}^2, \\
    \norm{x} \leq 1,
\end{gather}
and $\Distribution_{\tilde{x}}$ is a $\phi$-oversampling distribution for
$\Distribution_x$.

Assume Approximate Sample access to $x'$ and Approximate Query access to $x$. From
Lemma~\ref{lemma:approx-tvd}, the $\ell_1$ distance between $\Distribution_{\tilde{x}}$
and $\Distribution_{x'}$ is at most $C_1\epsilon/\sqrt{\phi} \RDefined \Delta$.
Furthermore, from a triangle inequality, $\norm{x'} \leq (1 + 2\Delta)\norm{\tilde{x}}$,
and from the definition of $\Distribution_{\tilde{x}}$, $\abs{x'(i)} = \norm{x'}
\sqrt{\Distribution_{x'}(i)}$.

From the Approximate Query access to $x$, one has that for a suitable constant $C_2$ and
in time $\Query(C_2\epsilon)$, and with probability $2/3$,
\begin{equation}
    \lforall[i \in [d]] \hat{x}_i : \abs{\hat{x}_i - x(i)} \leq C_2\epsilon.
\end{equation}
Let, for convenience of notation, $\epsilon_Q \Defined C_2\epsilon$.

Finally, for a choice of constant $C_3$, one has from Theorem~\ref{thm:smpl-estim} that
taking $\BigO[\epsilon^{-2}\phi^2 \log(d/\delta)]$ samples of $\Distribution_{x'}$ is
sufficient to estimate
\begin{equation}
    \lforall[i \in [d]] \hat{p}_i : \abs{\hat{p}_i - \Distribution_{x'}(i)} \leq C_3 \frac{\epsilon^2}{2\phi}
\end{equation}
with overall probability $1-\delta$. Again for convenience of notation we define
$\epsilon_P \Defined C_3\epsilon^2/\phi$.

With these choices, define the point estimator
\begin{equation}
    \widehat{x^\dagger y}(i) = \begin{cases}
        0                                                   & \text{if } \hat{p}_i \leq \frac{3}{2}\epsilon_P, \\
        \frac{1}{\hat{p}_i + \epsilon_P/2} \hat{x}_i^* y(i) & \text{otherwise.}
    \end{cases}
\end{equation}

Using the definitions above, the triangle inequality (possibly with double counting), and
H\"{o}lder inequalities, as well as the fact that $\abs{x(i)} \leq \abs{\tilde{x}(i)}$,
and $\ell_p$ norm inequalities, the bias of this estimator is bounded as
\begin{align}
    \abs{\Expectation_{i\sim\Distribution_{x'}} \widehat{x^\dagger y}(i) - x^\dagger y \mkern 3mu} \hspace*{-.5in} & \\
                                                                                                             & \leq \mkern 15mu \smashoperator{\sum_{i : {\Distribution_{x'}(i) \leq 2\epsilon_P}}} \mkern 20mu \abs{x(i)}\abs{y(i)} \mkern 10mu + \mkern-13mu \sum_{i : {\Distribution_{x'}(i) \geq \epsilon_P}} \mkern-2mu \abs{\mkern2mu \frac{\Distribution_{x'}(i)}{\hat{p}_i + \epsilon_P/2} \hat{x}_i^* y(i) - x(i)^* y(i) \mkern2mu}                                                                                                                                                                \\
                                                                                                             & \leq \begin{multlined}[t] \Big[\sum_{i: \Distribution_{x'}(i) \leq 2\epsilon_P} \abs{\tilde{x}(i) - x'(i)}\abs{y(i)} + \abs{x'(i)}\abs{y(i)} \mkern 10mu \Big] \\
                                                                                                                        {}+ \Big[\sum_{i:\Distribution_{x'}(i) \geq \epsilon_P} \abs{\frac{\Distribution_{x'}(i)}{\hat{p}_i + \epsilon_P/2} - 1}\abs{x(i)}\abs{y(i)} + \abs{\hat{x}_i - x(i)}\abs{y(i)} \\{}+ \abs{\frac{\Distribution_{x'}(i)}{\hat{p}_i + \epsilon_P/2} - 1} \abs{\hat{x}_i - x(i)} \abs{y(i)} \mkern8mu \Big] \mkern-18mu \end{multlined}                                                                                                                                                   \\
                                                                                                             & \leq \begin{multlined}[t] \Big[ \Delta \norm{y}_\infty + \mkern5mu \smashoperator{\sum_{i:{\Distribution_{x'}(i) \leq 2\epsilon_P}}} \mkern10mu \norm{x'}\sqrt{\Distribution_{x'}(i)}\abs{y(i)} \Big] \\{} + \Big[ \sum_{i:{\Distribution_{x'}(i) \geq \epsilon_P}} \frac{\epsilon_P}{\Distribution_{x'}(i)} \abs{x(i)} \abs{y(i)} + \epsilon_Q\abs{y(i)} + \frac{\epsilon_P}{\Distribution_{x'}(i)} \epsilon_Q \abs{y(i)} \mkern 8mu \Big] \end{multlined} \\
                                                                                                             & \leq \begin{multlined}[t] \Big[ \Delta \norm{y}_1 + 2\sqrt{\epsilon_P}(1 + 2\Delta) \norm{\tilde{x}}\norm{y}_1 \Big] \\ + \Big[ \sum_{i: \Distribution_{x'}(i) \geq \epsilon_P} \abs{\tilde{x}(i) - x'(i)}\abs{y(i)} + \abs{x'(i)}\abs{y(i)} + \epsilon_Q \norm{y}_1 + \epsilon_Q \norm{y}_1 \mkern8mu \Big] \end{multlined}                                                                             \\
                                                                                                             & \leq [C_1 + \sqrt{C_3}(1 + 2C_1)]\epsilon \norm{y}_1 + \Delta\norm{y}_1 \mkern2mu + \mkern2mu \smashoperator{\sum_{i:\Distribution_{x'}(i)\geq\epsilon_P}} \mkern10mu \norm{x'}\sqrt{\Distribution_{x'}(i)}\abs{y(i)} + 2\epsilon_Q\norm{y}_1                                                                                                                                                                                                                                              \\
                                                                                                             & \leq [C_1 + (1 + 2C_1)\sqrt{C_3} + 2C_2] \; \epsilon\norm{y}_1.
\end{align}

Likewise, we bound the variance of this estimator with the same techniques and a geometric
mean--quadratic mean inequality:
\begin{align}
    \Variance_{i\sim\Distribution_{x'}} \widehat{x^\dagger y}(i) & \leq \Expectation_{i\sim\Distribution_{x'}} \abs{\widehat{x^\dagger y}(i)}^2                                                                                                 \\
                                                                 & \leq \sum_{i:{\Distribution_{x'}(i)\geq\epsilon_P}} \frac{\Distribution_{x'}(i)}{\hat{p}_i + \epsilon_P/2} \frac{1}{\hat{p}_i + \epsilon_P/2} \abs{\hat{x}_i}^2 \abs{y(i)}^2 \\
                                                                 & \leq \sum_{i:{\Distribution_{x'}(i)\geq\epsilon_P}} \frac{1}{\Distribution_{x'}(i)} [\abs{x(i)} + \epsilon_Q]^2 \abs{y(i)}^2                                                 \\
                                                                 & \leq \sum_{i:{\Distribution_{x'}(i)\geq\epsilon_P}} \frac{1}{\Distribution_{x'}(i)} [\abs{x'(i)} + \abs{\tilde{x}(i) - x'(i)} + \epsilon_Q]^2 \abs{y(i)}^2                   \\
                                                                 & \leq 3\norm{x'}^2\norm{y}^2 + 3 \frac{1}{\epsilon_P} \Delta^2\norm{y}^2 + 3\frac{\epsilon_Q^2}{\epsilon_P} \norm{y}^2                                                        \\
                                                                 & \leq [3(1+2C_1)^2 + 3\frac{C_1^2}{C_3} + 3\frac{C_2^2}{C_3}]\;\phi\norm{y}^2.
\end{align}

Therefore we have a bias bounded as $\BigO(\epsilon\norm{y}_1)$, and a variance bounded
as $\BigO(\phi\norm{y}^2)$. It follows from the Chebyshev inequality that by taking the
average of $\BigO(\phi\epsilon^{-2}\delta^{-1})$ valid samples of the point estimator,
the resulting value deviates from $x^\dagger y$ at most $\epsilon\norm{y}_1$ with
probability at least $1-\delta$. The rest of the proof follows identically to that of
Lemma~\ref{lemma:inprod-asym}.

\section{Proof of \texorpdfstring{Remark~\ref{remark:inprod-sym-approx}}{Remark 29}}
\label{app:inprod-sym-approx}

Let $\epsilon \in [0,1)$, and $x$, $y$, $\tilde{x}$, $\tilde{y}$ such that for suitable
constants $C_1$, $C_2$, $C_3$: $\norm{x},\norm{y} \leq 1$, $\norm{\tilde{x}}^2 \leq
\phi\norm{x}^2$, $\norm{\tilde{y}}^2 \leq \phi \norm{y}^2$, and
\begin{alignat}{2}
    \norm{x' - \tilde{x}} & \leq \frac{C_1 \epsilon}{4\sqrt{\phi}} \norm{\tilde{x}}, & \qquad  \norm{y' - \tilde{y}} & \leq \frac{C_1 \epsilon}{4\sqrt{\phi}} \norm{\tilde{y}}.
\end{alignat}
Define, for convenience of notation, $\Delta \Defined C_1\epsilon/\sqrt{\phi}$.
Furthermore assume that $\Distribution_{\tilde{x}}$ is a $\phi$-oversampling distribution
for $\Distribution_x$, and $\Distribution_{\tilde{y}}$ is a $\phi$-oversampling
distribution for $y$. Note that the conditions imposed on $\norm{\tilde{x} - x'}$ and
$\norm{\tilde{y} - y'}$ imply that $\norm{x'} \leq (1+\Delta/4)\norm{\tilde{x}}$ and
likewise $\norm{y'} \leq (1+\Delta/4)\norm{\tilde{y}}$.

Let $\Distribution$ and $\AltDistribution$ be defined by their probability mass
functions; for all $i\in[d]$,
\begin{alignat}{2}
    \Distribution(i) & = \frac{\Distribution_{\tilde{x}}(i) + \Distribution_{\tilde{y}}(i)}{2}, & \qquad  \AltDistribution(i) & = \frac{\Distribution_{x'}(i) + \Distribution_{y'}(i)}{2}.
\end{alignat}
One may sample from $\Distribution$ by outputting uniformly randomly from either
$\Distribution_{\tilde{x}}$ or $\Distribution_{\tilde{y}}$, and from $\AltDistribution$
by outputting uniformly randomly from either $\Distribution_{x'}$ or
$\Distribution_{y'}$.

From Theorem~\ref{thm:smpl-estim}, one has that $\BigO(\epsilon^{-4}\phi^2\log d)$
samples of $\AltDistribution$ suffice to estimate $\AltDistribution(i)$ to error
$C_2\epsilon^2/(2\phi)$ with probability $8/9$. [After this, per the Theorem, reading an
estimate takes time $\BigO(\log d)$.] Let, for ease of notation $\gamma \Defined
C_2\epsilon^2/\phi$, such that each estimate of $\AltDistribution(i)$, say, $\hat{q}_i$
is accurate to error $\gamma/2$ (with error probability $1/9$).

Likewise, if $\epsilon_Q \Defined C_3\epsilon^2/\phi$, one has from the Approximate Query
access that, for all $i=1, \ldots, d$, it is possible to compute an estimate of $x(i)$
and $y(i)$ --- say, $\hat{x}_i$ and $\hat{y}_i$, respectively --- to absolute error
$\epsilon_Q$ and with error probability $1/3$ in time $\Query(\epsilon_Q)$.

To relate $\abs{x(i)}$ to $\AltDistribution(i)$, we will make use of the following bound,
following from the definitions above and a triangle inequality:
\begin{align}
    \abs{x(i)} \leq \abs{\tilde{x}(i)} & \leq \abs{x'(i)} + \abs{x'(i) - \tilde{x}(i)} \leq \norm{x'} \sqrt{\Distribution_{x'}(i)} + \abs{x'(i) - \tilde{x}(i)} \\&\leq (1 + \Delta/4)\norm{\tilde{x}} \sqrt{2\AltDistribution(i)} + \abs{x'(i) - \tilde{x}(i)} \leq (1 + 2\Delta)\sqrt{2\phi} \sqrt{\AltDistribution(i)} + \abs{x'(i) - \tilde{x}(i)}.
\end{align}
An analogous bound holds for $\abs{y(i)}$. We will also make use of the following fact,
which makes use of an arithmetic mean--geometric mean inequality:
\begin{equation}
    \abs{x(i)}\abs{y(i)} \leq \frac{\abs{\tilde{x}(i)}}{\norm{\tilde{x}}} \frac{\abs{\tilde{y}(i)}}{\norm{\tilde{y}}} \norm{\tilde{x}}\norm{\tilde{y}} \leq \phi \Distribution(i) \leq \phi[\AltDistribution(i) + \abs{\Distribution(i) - \AltDistribution(i)}].
\end{equation}
Finally, as elsewhere, we will make use of the fact that if $\hat{q}_i$ is an estimate of
$\AltDistribution(i)$ to error $\epsilon_P/2$, then $(\hat{q}_i + \epsilon_P/2)^{-1} \leq
\AltDistribution(i)^{-1}$.

Now, define the estimator
\begin{equation}
    \widehat{x^\dagger y}(i) = \begin{cases}
        0                                                    & \text{if } \hat{q}_i \leq \frac{3}{2} \gamma \\
        \frac{1}{\hat{q}_i + \gamma/2} \hat{x}_i^* \hat{y}_i & \text{otherwise.}
    \end{cases}
\end{equation}
The bias of this estimator (with respect to $x^\dagger y$) is bounded as
\begin{equation}
    \abs{\Expectation_{i\sim\AltDistribution} \widehat{x^\dagger y} - x^\dagger y \mkern 2mu } \leq \sum_{i:\AltDistribution(i)\leq2\gamma} \abs{x(i)} \abs{y(i)} \mkern 5mu + \mkern-5mu \sum_{i:\AltDistribution(i)\geq\gamma} \abs{\frac{\AltDistribution(i)}{\hat{q}_i + \gamma/2} \hat{x}_i^* \hat{y}_i - x(i)^* y(i)}
\end{equation}
as follows from a triangle inequality and some potential double counting. We further
bound these two terms separately, making use of the Cauchy-Schwarz inequality and the
identities above:
\begin{align}
    \sum_{i:\AltDistribution(i)\leq2\gamma} \abs{x(i)}\abs{y(i)}                                                                      & \leq (1 + \frac{\Delta}{4}) 2 \sqrt{\phi\gamma}\norm{y}_1 + \mkern -15mu \sum_{i:\AltDistribution(i)\leq2\gamma} \mkern -5mu \abs{x'(i) - \tilde{x}(i)} \abs{y(i)}                                                                                                                                                                         \\
                                                                                                                                      & \leq 2(1 + \frac{\Delta}{4}) \sqrt{\gamma \phi} \norm{y}_1 + \norm{\tilde{x} - x'} \norm{y} \leq 2(1 + \frac{\Delta}{4})\sqrt{\gamma\phi} \norm{y}_1 + \frac{\Delta}{4} \sqrt{\phi}.
    \intertext{But, by a symmetry argument, we have}
    \sum_{i:\AltDistribution(i)\leq2\gamma} \abs{x(i)}\abs{y(i)}                                                                      & \leq 2(1 + \frac{\Delta}{4}) \sqrt{\gamma\phi} \min(\norm{x}_1, \norm{y}_1) + \frac{\Delta}{4} \sqrt{\phi}.
\end{align}
The second term is bounded as
\begin{align}
    \sum_{i:\AltDistribution(i)\geq\gamma} \abs{\frac{\AltDistribution(i)}{\hat{q}_i + \gamma/2} \hat{x}_i^* \hat{y}_i - x(i)^* y(i)} & \leq \sum_{i:\AltDistribution(i)\geq\gamma} \begin{multlined}[t] \abs{\frac{\AltDistribution(i)}{\hat{q}_i + \gamma/2} - 1} \abs{x(i)} \abs{y(i)} + \abs{\hat{x}_i^* \hat{y}_i - x(i)^* y(i)} \\+ \abs{\frac{\AltDistribution(i)}{\hat{q}_i + \gamma/2} - 1} \abs{\hat{x}_i^* \hat{y}_i - x(i)^* y(i)} \end{multlined} \\
                                                                                                                                      & \hskip-3cm \leq \sum_{i:\AltDistribution(i)\geq\gamma} \frac{\gamma}{\AltDistribution(i)} \abs{x(i)}\abs{y(i)} + 2[\epsilon_Q \abs{x(i)} + \epsilon_Q \abs{y(i)} + \epsilon_Q^2]                                                                                                                                                  \\
                                                                                                                                      & \hskip-3cm \leq (1 + \frac{\Delta}{4}) \sqrt{2\phi\gamma} \min(\norm{x}_1, \norm{y}_1) + \frac{\Delta}{4} \sqrt{\phi} + \Big[4\sqrt{\phi} \frac{\epsilon_Q}{\sqrt{\gamma}} + 2\frac{\epsilon_Q}{\gamma}\epsilon_Q\Big].
\end{align}

Therefore, the bias is overall bounded as
\begin{equation}
    (2\sqrt{2C_2}) \epsilon\min(\norm{x}_1, \norm{y}_1) + (\frac{C_1}{4} + \frac{4C_3}{\sqrt{C_2}} + \frac{2C_3^2}{C_2}) \epsilon = \BigO\{\epsilon\,[1 + \min(\norm{x}_1, \norm{y}_1)]\}.
\end{equation}

We now bound the variance of the estimator:
\begin{align}
    \Variance_{i\sim\AltDistribution} \widehat{x^\dagger y}(i) & \leq \Expectation_{i\sim\AltDistribution} \abs{\widehat{x^\dagger y}(i)}^2 \leq \sum_{i:\AltDistribution(i)\geq\gamma} \frac{1}{\AltDistribution(i)} \abs{\hat{x}_i}^2 \abs{\hat{y}_i}^2 \leq 4\sum_{i:\AltDistribution(i)\geq\gamma} \frac{1}{\AltDistribution(i)} (\abs{x(i)}^2 + \epsilon_Q^2) (\abs{y(i)}^2 + \epsilon_Q^2)                                    \\
                                                               & \leq \mkern3mu 4\mkern-8mu\sum_{i:\AltDistribution(i)\geq\gamma}\frac{1}{\AltDistribution(i)}\abs{x(i)}^2 \abs{y(i)}^2 + \mkern3mu 4\mkern-8mu\sum_{i:\AltDistribution(i)\geq\gamma}\frac{1}{\AltDistribution(i)}\epsilon_Q^2(\abs{x(i)}^2 + \abs{y(i)}^2) + \mkern3mu 4\mkern-8mu\sum_{i:\AltDistribution(i)\geq\gamma}\frac{1}{\AltDistribution(i)}\epsilon_Q^4.
\end{align}
Again further bounding term by term:
\bgroup\def\Tab{\hphantom{\sum_{i:\AltDistribution(i)\geq\gamma}\frac{1}{\AltDistribution(i)}\abs{x(i)}^2 \abs{y(i)}^2}}
\begin{align}
     & \sum_{i:\AltDistribution(i)\geq\gamma}\frac{1}{\AltDistribution(i)}\abs{x(i)}^2 \abs{y(i)}^2 \leq 4\phi \sum_{i:\AltDistribution(i)\geq\gamma} \frac{\Distribution(i)}{\AltDistribution(i)} \abs{x(i)}\abs{y(i)} \leq 4\phi + 4\phi \sum_{i:\AltDistribution(i)\geq\gamma} \frac{\abs{\Distribution(i) - \AltDistribution(i)}}{\AltDistribution(i)} \abs{x(i)}\abs{y(i)} \\
     & \Tab \leq 4\phi + 4\phi^2 \Delta + 4\phi^2 \sum_{i:\AltDistribution(i)\geq\gamma} \abs{\Distribution(i) - \AltDistribution(i)}^2 \frac{1}{\AltDistribution(i)} \leq 4\phi + 4\phi^2 \Delta + 4\phi^2 \frac{\Delta^2}{\gamma}                                                                                                                                             \\
     & \Tab \leq 4(1 + C_1 + \frac{C_1^2}{C_2}) \phi^2;                                                                                                                                                                                                                                                                                                                         \\
     & 4 \mkern-8mu \sum_{i:\AltDistribution(i)\geq\gamma}\frac{1}{\AltDistribution(i)}\epsilon_Q^2(\abs{x(i)}^2 + \abs{y(i)}^2) \leq 8\phi\sum_{i:\AltDistribution(i)\geq\gamma} \frac{\Distribution(i)}{\AltDistribution(i)} \epsilon_Q^2 \leq 8\phi \frac{\epsilon_Q^2}{\gamma} = \frac{8 C_3^2}{C_2};                                                                       \\
     & 4\mkern-8mu\sum_{i:\AltDistribution(i)\geq\gamma}\frac{1}{\AltDistribution(i)}\epsilon_Q^4 \leq 4\frac{\epsilon_Q^4}{\gamma^2} = \frac{256\, C_3^4}{C_2^2}.                                                                                                                                                                                                              \\
\end{align}
\egroup

Thus, the variance is bounded as $\BigO(\phi^2)$, and the rest of the proof proceeds identically to that of Lemma~\ref{lemma:inprod-sym}.

\section{Approximate sample and query for block encoded matrices}
\label{app:asq-bemat}

\begin{definition}{\textit{Approximate $\phi$-oversample and query} access --- matrices}
    For $\phi\geq 1$, say that one has Approximate $\phi$-Oversample and Query
    ($\ASQ[\phi]$) access to a matrix $A \in \Complex^{n\times m}$ if, for a matrix
    $\tilde{A}$, satisfying $\norm{\tilde{A}}_F^2 \leq \phi\norm{A}_F^2$ and
    $\abs{\tilde{A}(i, j)}^2 \geq \abs{A(i, j)}$, for all $i,j \in [m]\times[n]$, one has
    query access to every column of $A$; AS access (cf.~Def.~\ref{def:as}) to every column
    of $\tilde{A}$; and one has AS access to the vector of column norms of $\tilde{A}$,
    i.e., $\tilde{a} \in \Real^m$ with entries
    $$ \tilde{a}(j) = \sqrt{\sum_{i\in[n]}\abs{\tilde{A}(i,j)}^2} $$
    with one-sided error probability at most $1/3$ and runtime $\textbf{s}_A$.
\end{definition}

Let $U$ be a $\QubitsM$-block encoding of a $2^\QubitsN\times2^\QubitsN$ matrix $A$
(cf.~section~\ref{sec:block-enc}). Given access to singly coherent
applications of both $U$ and $U^\dagger$, we claim that it is possible to satisfy $\ASQ[\phi]$ access
to $A$. The first two conditions, as outlined above, are easy to satisfy: if one prepares
$U\ket{0}^{\otimes\QubitsM}\ket{j}$, then one has a block encoding of the $j$th column of
$A$, and we have seen that this is sufficient to satisfy ASQ access to the $j$th column
of $A$. Therefore, it remains to show that the third condition can also be satisfied.
Consider the following procedure:
\begin{enumerate}
    \item Choose uniformly randomly a string of $\QubitsN$ bits.
    \item Prepare the corresponding computational basis state.
    \item Run the computational basis state through $U^\dagger$ (with the block register
          suitably initialized to all-zeroes).
    \item Postselecting for the correct block, output the readout of the remaining
          register.
\end{enumerate}

One finds that the output is the index of a column of $A$ (say, $k$), with probability
$\norm{A(*,k)}^2/\norm{A}_F^2$, where $A(*,k)$ is the $k$th column of $A$. Moreover, for the procedure to succeed with probability $2/3$, one
finds that $\BigO(d \norm{A}_F^2)$ rounds of post-selection are expected to be required.

\end{document}